\documentclass[conference]{IEEEtran}
\usepackage[latin9]{inputenc}
\usepackage{colortbl}
\usepackage{array}
\usepackage{pifont}
\usepackage{float}
\usepackage{booktabs}
\usepackage{bbding}
\usepackage{multirow}
\usepackage{amsmath}
\usepackage{amsthm}
\usepackage{amssymb}
\usepackage{graphicx}

\makeatletter

\providecommand{\tabularnewline}{\\}
\floatstyle{ruled}
\newfloat{algorithm}{tbp}{loa}
\providecommand{\algorithmname}{Algorithm}
\floatname{algorithm}{\protect\algorithmname}

\IEEEoverridecommandlockouts

\usepackage{cite}
\usepackage{amsfonts}
\usepackage{amsthm}

\usepackage{textcomp}
\usepackage{xcolor}
\def\BibTeX{{\rm B\kern-.05em{\sc i\kern-.025em b}\kern-.08em
   T\kern-.1667em\lower.7ex\hbox{E}\kern-.125emX}}

\providecommand{\tabularnewline}{\\}
\floatstyle{ruled}
\newfloat{algorithm}{tbp}{loa}
\providecommand{\algorithmname}{Algorithm}
\floatname{algorithm}{\protect\algorithmname}

\usepackage{algorithm}
\usepackage{algpseudocode}
\algnewcommand{\Linecomment}[1]{\Statex \(\triangleright\) #1}
\algrenewcommand\algorithmicrequire{\textbf{Input:}}
\algrenewcommand\algorithmicensure{\textbf{Output:}}
\usepackage[justification=centering]{caption}

\theoremstyle{plain}

\newtheorem{lemma}{Lemma}
\newtheorem{Corollary}{Corollary}

\makeatother

\theoremstyle{plain}
\newtheorem{thm}{\protect\theoremname}
\providecommand{\theoremname}{Theorem}

\begin{document}
	\title{An $O\left(K\right)$-Approximation Algorithm for\\
		Scheduling Coflows in $K$-Core Optical \\
		Circuit Switching Networks}
	\author{{\normalsize Xin Wang$^{a}$, Hong Shen$^{a}$, Hui Tian$^{b}$, Ye
			Tao$^{a}$}\\
		{\normalsize}\\
		{\normalsize$^{a}$School of Engineering and Technology, Central Queensland
			University, Australia}\\
		{\normalsize$^{b}$School of Information and Communication Technology,
			Griffith University, Australia}}
	\maketitle
\begin{abstract}
	Motivated by heterogeneous parallel networks
	(HPNs) that aggregate bandwidth across independent
	cores, this paper studies coflow scheduling in
	multi-core optical circuit switching (OCS) networks
	under the \textit{not-all-stop} reconfiguration model.
	The objective is to minimize the total weighted
	coflow completion time (CCT) by jointly determining
	flow allocation across heterogeneous cores and
	circuit scheduling within each core.
	Flow allocation must account for both transmission
	workload and reconfiguration overhead, while circuit
	scheduling must respect port exclusivity.
	Failing to coordinate flow allocation and circuit
	scheduling can lead to unbalanced core workloads
	and prolonged CCTs.
	
	We propose an approximation algorithm that combines
	linear programming (LP)-guided coflow ordering,
	greedy flow allocation, and circuit scheduling within
	groups defined by LP completion times.
	For heterogeneous $K$-core OCS networks, the algorithm
	achieves approximation ratios of $32K$ and $32K+2$
	for zero and arbitrary release times, respectively.
	These guarantees depend only on the number of OCS cores.
	The same framework applies to heterogeneous
	$H$-core electrical packet switching (EPS) networks
	through a transmission-only variant, yielding
	approximation guarantees of $16H$ and $16H+2$
	for zero and arbitrary release times, respectively.
	Trace-driven simulations using a Facebook workload
	show that the proposed algorithm achieves lower total
	weighted CCT than the evaluated baselines.
\end{abstract}


\section{Introduction}
Distributed computing frameworks, such as MapReduce \cite{mapreduce},
Spark \cite{spark}, and Dryad \cite{dryad}, typically execute jobs
through multiple computation and communication stages separated by
synchronization barriers. A subsequent stage can begin only after
all parallel flows in the current communication stage have completed.
The \textit{coflow} \cite{networking} captures this application-level
dependency by grouping related flows into a single scheduling entity.
A coflow completes only when its last flow finishes,
so optimizing individual flow completion time (FCT)
alone is insufficient to optimize job-level performance. The coflow completion time (CCT) can therefore
be used to measure application-level communication performance.

Prior research on coflow scheduling \cite{theoretical5,literature6,Baraat,decentralized1,Aalo,wang2023online,wang2023efficient,literature9,improved,literature32}
has primarily adopted the single-core electrical packet switching
(EPS) model, where the data center network (DCN) is abstracted as
a single non-blocking switch fabric with full bisection bandwidth.
However, as data center traffic continues to grow, pure packet-switched
architectures face increasing challenges in scalability, cost, and
power consumption. To mitigate these issues, optical circuit switching
(OCS) fabrics have been incorporated into DCNs, providing high-capacity
circuits for large-volume traffic and thereby improving communication
efficiency. Several coflow scheduling schemes \cite{omco,reco,regularization,sunflow,zhang2020minimizing}
have been developed for single-core OCS networks. More recently, hybrid
EPS-OCS architectures have been studied for coflow scheduling \cite{wang2024scheduling,co-scheduler,wang2025optimal,ONS,wang2026coflow},
where packet and circuit resources coexist and are jointly optimized.

Nevertheless, the single-core abstraction is increasingly insufficient
for modern DCN architectures. Modern data centers are evolving toward parallel designs,
where multiple network cores operate concurrently
to increase aggregate bandwidth \cite{cisco2016_40g,cisco2016gci}.
In practice, network cores from different generations
may coexist, forming heterogeneous parallel networks
(HPNs), in which multiple independent cores share
the same set of hosts \cite{huang2020weaver}. Motivated by this trend, coflow scheduling
has been studied in multi-core EPS networks to exploit parallel packet-switched
capacity \cite{huang2020weaver,chen2023efficient}. Parallelism is
also emerging in circuit-switched fabrics. For example, Google's Jupiter architecture uses an optical switching layer to interconnect aggregation blocks, enabling incremental capacity expansion and dynamic topology reconfiguration \cite{poutievski2022jupiter}. However, coflow scheduling in multi-core OCS networks has received limited attention.

The key difficulty lies in the tight coupling between inter-core flow
allocation and intra-core circuit scheduling: allocation determines
the workload assigned to each heterogeneous OCS core, while each core
must schedule its assigned traffic as port-exclusive circuits and
account for reconfiguration delay. Port exclusivity restricts each
ingress or egress port to participating in at most one circuit at
a time, and each circuit reconfiguration incurs a non-negligible delay.
OCS reconfiguration mechanisms are generally classified into the \textit{all-stop}
and \textit{not-all-stop} models. In the \textit{all-stop} model,
reconfiguration is synchronous, i.e., any circuit change pauses all
transmissions on the corresponding core. In contrast, in the \textit{not-all-stop}
model, reconfiguration is asynchronous, i.e., only the reconfigured
ports are interrupted, while unaffected circuits continue transmitting.
This paper focuses on the \textit{not-all-stop} model, which introduces greater scheduling complexity
due to asynchronous circuit updates.

Compared with single-core OCS scheduling, the multi-core setting
additionally requires flow allocation across heterogeneous cores.
Compared with multi-core EPS scheduling, it must also account
for circuit-level port exclusivity and reconfiguration delay.
We develop an efficient approximation algorithm for multi-coflow
scheduling in multi-core OCS networks under the
\textit{not-all-stop} reconfiguration model.
The main contributions of this paper are as follows:

\begin{itemize}
	\item We formulate a linear programming (LP) relaxation
	for minimizing the total weighted CCT.
	The relaxation incorporates coflow release times and
	port capacity constraints on cumulative transmission
	and reconfiguration workloads.
	Its optimal objective value provides a lower bound
	on the minimum total weighted CCT.
	The LP completion times determine the coflow order
	for greedy flow allocation and the geometric groups
	for circuit scheduling.
	
	\item We prove that the proposed algorithm achieves approximation
	ratios of $32K$ and $32K+2$ for zero and arbitrary release times,
	respectively, in heterogeneous $K$-core OCS networks.
	Compared with the prior worst-case approximation bound of
	$O\!\left(M\frac{w_{\max}}{w_{\min}}
	\max\{K,\tau_{\max}\}\right)$
	\cite{wang2026scheduling}, our guarantees depend only on $K$.
	Here, $M$ is the number of input coflows,
	$w_{\max}/w_{\min}$ is the ratio between the maximum
	and minimum coflow weights, and $\tau_{\max}$ is the
	maximum number of nonzero flows incident to any ingress
	or egress port within a single coflow, maximized over
	all $M$ input coflows.
	
	\item We derive a transmission-only variant for heterogeneous
	$H$-core EPS networks.
	This variant removes the reconfiguration constraints
	from the LP relaxation and the reconfiguration workload
	term from flow allocation.
	It retains grouped scheduling with zero reconfiguration delay
	and achieves approximation ratios of $16H$ and $16H+2$
	for zero and arbitrary release times, respectively.
	
	\item Trace-driven simulations using a Facebook workload show that the proposed algorithm reduces total weighted CCT relative to the evaluated baselines across a range of network and workload settings.
\end{itemize}
The remainder of this paper is organized as follows.
Section~\ref{sec:Related-Work} reviews related research
and provides a comparative analysis.
Section~\ref{sec:Model-and-Problem} introduces the system model
and problem formulation.
Section~\ref{sec:Multiple-Coflow-Scheduling} describes the
proposed multi-coflow scheduling algorithm and its
theoretical performance guarantees.
Section~\ref{sec:Experimental-Evaluations} presents the
experimental evaluation using a Facebook trace.
Section~\ref{sec:Conclusions} concludes the paper.

\section{Related Work\label{sec:Related-Work}}
Coflow scheduling has been studied under several DCN switching models. Early works mainly considered single-core EPS networks, while later studies extended the problem to single-core OCS (including hybrid EPS-OCS) architectures under both \textit{all-stop} and \textit{not-all-stop} reconfiguration models. Research has also explored multi-core EPS and OCS networks. This section reviews the most relevant studies from these architectural perspectives and summarizes them in Table \ref{tab:RW_table}.

\begin{table*}[tp]
\centering{}{}{}\vspace*{0.8\baselineskip}
 \caption{COMPARISON AMONG RELATED WORK}
\label{tab:RW_table} %
\begin{tabular}{c|c|c|c|c|c}
\hline 
\multirow{2}{*}{Works} & \multicolumn{2}{c|}{Single-Core} & \multicolumn{2}{c|}{Multi-Core} & \multirow{2}{*}{Performance Guarantee}\tabularnewline
\cline{2-5}
 & \multirow{1}{*}{EPS-Enabled} & OCS-Enabled & EPS-Enabled & OCS-Enabled & \tabularnewline
\hline 
Varys \cite{literature6}, Baraat \cite{Baraat}, D-CAS \cite{decentralized1} & \CheckmarkBold{} & \ding{56} & \multicolumn{1}{c|}{\ding{56}} & \ding{56} & \ding{56}\tabularnewline
\hline 
Qiu \textit{et al. }\cite{literature9}, Khuller \textit{et al.} \cite{literature10},
Shafiee \textit{et al.} \cite{improved} & \CheckmarkBold{} & \ding{56} & \ding{56} & \ding{56} & \CheckmarkBold{}\tabularnewline
\hline 
OMCO \cite{omco} & \ding{56} & \CheckmarkBold{} & \ding{56} & \ding{56} & \ding{56}\tabularnewline
\hline 
Sunflow \cite{sunflow}, Reco-Sin \cite{reco}, Reco-Mul+ \cite{regularization},
GOS \cite{zhang2020minimizing} & \ding{56} & \CheckmarkBold{} & \ding{56} & \ding{56} & \CheckmarkBold{}\tabularnewline
\hline 
Co-scheduler \cite{co-scheduler}, ONS \cite{ONS} & \CheckmarkBold{} & \CheckmarkBold{} & \multirow{1}{*}{\ding{56}} & \ding{56} & \ding{56}\tabularnewline
\hline 
Wang \textit{et al.} \cite{wang2024scheduling}, Wang \textit{et al.}
\cite{wang2025optimal}, Wang \textit{et al.}
\cite{wang2026coflow}& \CheckmarkBold{} & \CheckmarkBold{} & \ding{56} & \ding{56} & \CheckmarkBold{}\tabularnewline
\hline 
Weaver \cite{huang2020weaver}, Chen \cite{chen2023scheduling}, Chen
\cite{chen2023efficient} & \ding{56} & \ding{56} & \multirow{1}{*}{\CheckmarkBold{}} & \ding{56} & \CheckmarkBold{}\tabularnewline
\hline 
Wang \textit{et al. }\cite{wang2026scheduling} & \ding{56} & \ding{56} & \ding{56} & \CheckmarkBold{} & \CheckmarkBold{}\tabularnewline
\hline
\textbf{Our Work} & \ding{56} & \ding{56} & \CheckmarkBold{} & \CheckmarkBold{} & \CheckmarkBold{}\tabularnewline
\hline 
\end{tabular}{}{}
\end{table*}

\subsection{Coflow Scheduling in Single-Core EPS Networks}
Varys \cite{literature6} proposed two heuristics, smallest-effective-bottleneck-first (SEBF) and minimum-allocation-for-desired-duration (MADD), to greedily schedule coflows based on bottleneck completion time. In decentralized settings, Baraat \cite{Baraat} addressed head-of-line blocking through multiplexing, while D-CAS \cite{decentralized1} focused on decentralized coflow scheduling. Aalo \cite{Aalo} and NC-DRF \cite{NC-DRF} are information-agnostic coflow schedulers that prioritize efficiency and fairness, respectively. More recently, Wang et al. \cite{wang2023online} proposed a DRL-based online multi-stage job scheduling framework. Furthermore, Wang et al. \cite{wang2023efficient} combined limited multiplexing with DRL to improve efficiency while maintaining fairness. However, these methods do not provide approximation guarantees for minimizing the total weighted CCT.

At the theoretical level, several approximation results have been established for coflow scheduling. Qiu et al. \cite{literature9} proposed a deterministic \(\frac{67}{3}\)-approximation algorithm to minimize the total weighted CCT. Khuller et al. \cite{literature10} modeled the problem as a concurrent open-shop problem and derived a 12-approximation algorithm. Shafiee et al. \cite{improved} subsequently improved this bound to 5 using an LP-based approach, while Wang et al. \cite{literature32} achieved a 2-approximation without solving an LP. Most of these works focus on independent coflows and do not consider precedence dependencies within multi-stage jobs. Tian et al. \cite{literature13} studied dependent coflows in multi-stage jobs and established a \((2N+1)\)-approximation for minimizing the total weighted job completion time (JCT), where \(N\) is the number of hosts. Subsequently, Shafiee et al. \cite{theoretical5} developed an \(O\!\left(\frac{\mu\log N}{\log\log N}\right)\)-approximation algorithm for the same objective, where \(\mu\) denotes the maximum number of coflows in a job.

\subsection{Coflow Scheduling in Single-Core OCS Networks}
Research on coflow scheduling in single-core OCS-based architectures, including pure OCS and hybrid EPS-OCS networks, remains relatively limited. Since OCS reconfiguration mechanisms are commonly categorized into the \textit{all-stop} and \textit{not-all-stop} models, we review the existing literature under both settings.

\subsubsection{\textit{All-Stop} Reconfiguration Model}
OMCO \cite{omco} proposed the first online algorithm for multi-coflow scheduling in single-core pure OCS networks. Reco-Sin \cite{reco} and Reco-Mul+ \cite{regularization} established the first approximation guarantees for single-coflow and multi-coflow scheduling in this setting, with approximation ratios of 2 and \(8M\), respectively, where \(M\) denotes the number of coflows. These methods rely on Birkhoff-von Neumann (BvN) decomposition \cite{BvN}. Wang et al. \cite{wang2024scheduling} further developed approximation algorithms with provable guarantees for both single- and multi-coflow scheduling in single-core hybrid EPS-OCS networks.

\subsubsection{\textit{Not-All-Stop} Reconfiguration Model}
Sunflow \cite{sunflow} proposed a constant-factor approximation algorithm for single-coflow scheduling in single-core pure OCS networks, together with a heuristic for multi-coflow scheduling. Subsequently, Zhang et al. \cite{zhang2020minimizing} proposed GOS, a 4-approximation algorithm for multi-coflow scheduling in the same setting. In hybrid EPS-OCS networks, Co-scheduler \cite{co-scheduler} considered both optical-electrical hybrid switching and coflow structures, but did not provide formal performance guarantees. ONS \cite{ONS} proposed an online heuristic algorithm for minimizing the total CCT in single-core hybrid EPS-OCS networks, also without theoretical bounds. Wang et al. \cite{wang2025optimal} further studied hybrid EPS-OCS scheduling and derived a relative approximation guarantee with respect to a given pure-OCS schedule. More recently, Wang et al. \cite{wang2026coflow} established a global approximation guarantee for single-coflow scheduling in hybrid EPS-OCS networks under \textit{not-all-stop} reconfiguration and proposed a heuristic extension for multiple coflows.

\subsection{Coflow Scheduling in Multi-Core EPS Networks}
In recent years, coflow scheduling in multi-core EPS networks has
attracted increasing attention. Weaver \cite{huang2020weaver} studied
single-coflow scheduling in heterogeneous parallel networks (HPNs)
and proposed an $O\left(H\right)$-approximation algorithm, where
$H$ is the number of EPS cores. Chen \cite{chen2023scheduling} further
investigated multi-coflow scheduling in HPNs and derived an $O\left(\frac{\log H}{\log\log H}\right)$-approximation
algorithm. For identical parallel networks, Chen \cite{chen2023efficient}
developed coflow-level approximation algorithms with approximation
ratios of $\left(4H+1\right)$ and $4H$ for arbitrary and zero release
times, respectively.

\subsection{Coflow Scheduling in Multi-Core OCS Networks}
Compared with the above models, theoretical results for multi-core
OCS coflow scheduling remain limited.
To the best of our knowledge, Wang et al. \cite{wang2026scheduling}
were the first to study multi-coflow scheduling in multi-core OCS
networks under \textit{not-all-stop} reconfiguration with zero
release times.
They established an
$O\!\left(M\frac{w_{\max}}{w_{\min}}
\max\{K,\tau_{\max}\}\right)$-approximation guarantee
for minimizing the total weighted CCT.
Here, $M$ is the number of input coflows,
$K$ is the number of OCS cores, and $w_{\max}$ and $w_{\min}$
are the maximum and minimum coflow weights, respectively.
The parameter $\tau_{\max}$ is the maximum number of nonzero
flows incident to any ingress or egress port within a single
coflow, maximized over all $M$ input coflows.
This bound depends on both the number of cores and
workload parameters.
In contrast, this paper develops an LP-guided scheduling
framework whose approximation guarantees depend only
linearly on $K$.

\section{System Model and Problem Formulation\label{sec:Model-and-Problem}}
This section introduces the system model, including the network architecture
and traffic abstraction. Based on this model, we formally define the multi-coflow
scheduling problem in heterogeneous multi-core OCS networks
under the \textit{not-all-stop} reconfiguration model
and analyze its computational complexity. The main notation used throughout
this paper is summarized in Table \ref{tab:notation1}.

\begin{table*}[t]
	\centering
	\vspace*{0.6\baselineskip}
	\caption{Mathematical Notation}
	\label{tab:notation1}
	\small
	\setlength{\tabcolsep}{12pt}
	\begin{tabular}{l>{\raggedright}p{0.70\linewidth}}
		\toprule
		Notation & Definition\tabularnewline
		\midrule
		
		$\mathcal C$
		& The set of coflows\tabularnewline
		
		$\mathcal M,M$
		& The index set and number of coflows, respectively,
		where $M=|\mathcal M|$\tabularnewline
		
		$\mathcal K,K$
		& The index set and number of parallel OCS cores,
		respectively, where $K=|\mathcal K|$\tabularnewline
		
		$N$
		& The number of ingress ports per core,
		equal to the number of egress ports\tabularnewline
		
		$\mathcal I,\mathcal J$
		& The index sets of ingress and egress ports, respectively,
		where $|\mathcal I|=|\mathcal J|=N$\tabularnewline
		
		$p$
		& Any port in $\mathcal I\cup\mathcal J$\tabularnewline
		
		$r^k,r_{\max}$
		& The transmission rate per port on core $k$
		and the maximum rate across cores, respectively
		\tabularnewline
		
		$R$
		& The aggregated port transmission rate across all cores,
		i.e., $R=\sum_{k\in\mathcal K}r^k$\tabularnewline
		
		$\delta$
		& The reconfiguration delay
		\tabularnewline
		
		$C_m$
		& The $m$-th coflow, where $1\le m\le M$\tabularnewline
		
		$\mathcal F_m$
		& The set of nonzero flows in $C_m$\tabularnewline
		
		$D_m,D_m^k$
		& The demand matrix of $C_m$ and its portion allocated
		to core $k$, respectively\tabularnewline
		
		$f_m(i,j),f_m^k(i,j)$
		& The flow from ingress port $i$ to egress port $j$
		in $C_m$ and its subflow transmitted on core $k$,
		respectively\tabularnewline
		
		$d_m(i,j),d_m^k(i,j)$
		& The data sizes of $f_m(i,j)$ and $f_m^k(i,j)$,
		respectively\tabularnewline
		
		$w_m,a_m$
		& The weight and release time of $C_m$,
		respectively\tabularnewline
		
		$t_m^k(i,j)$
		& The start time of circuit reconfiguration for $f_m^k(i,j)$
		\tabularnewline
		
		$T_m,T_m^\ast$
		& The completion times of $C_m$ in a feasible schedule
		and an optimal schedule, respectively\tabularnewline
		
		$\rho_{m,p},\rho_{m,p}^k$
		& The loads incident to port $p$ in $D_m$ and $D_m^k$,
		respectively\tabularnewline
		
		$\tau_{m,p},\tau_{m,p}^k$
		& The numbers of nonzero flow entries incident to port $p$
		in $D_m$ and $D_m^k$, respectively\tabularnewline
		
		$\rho_m,\rho_m^k$
		& The maximum port loads in $D_m$ and $D_m^k$,
		respectively\tabularnewline
		
		$\tau_m,\tau_m^k$
		& The maximum numbers of nonzero flow entries incident
		to any port in $D_m$ and $D_m^k$,
		respectively\tabularnewline
		
		$x_{m,m'}$
		& The binary ordering variable indicating whether $C_m$
		precedes $C_{m'}$ in the coflow completion order
		\tabularnewline
		
		$\widetilde T_m$
		& The LP completion time of coflow $C_m$\tabularnewline
		
		$L_m$
		& The prefix set of the first $m$ coflows in the reindexed
		LP-guided order, i.e., $L_m=\{1,\ldots,m\}$
		\tabularnewline
		
		$\rho_{1:m,p}$
		& The cumulative load incident to port $p$ for coflows
		in $L_m$, i.e.,
		$\rho_{1:m,p}\triangleq\sum_{\ell=1}^m\rho_{\ell,p}$
		\tabularnewline
		
		$\rho_{1:m,p}^k$
		& The cumulative load incident to port $p$ on core $k$
		for coflows in $L_m$, i.e.,
		$\rho_{1:m,p}^k\triangleq\sum_{\ell=1}^m\rho_{\ell,p}^k$
		\tabularnewline
		
		$\tau_{1:m,p}$
		& The cumulative number of nonzero flows incident to port $p$
		for coflows in $L_m$, i.e.,
		$\tau_{1:m,p}\triangleq\sum_{\ell=1}^m\tau_{\ell,p}$
		\tabularnewline
		
		$\tau_{1:m,p}^k$
		& The cumulative number of nonzero subflows incident to port $p$
		on core $k$ for coflows in $L_m$, i.e.,
		$\tau_{1:m,p}^k\triangleq\sum_{\ell=1}^m\tau_{\ell,p}^k$
		\tabularnewline
		
		$\rho_{1:m}$
		& The maximum cumulative load for coflows in $L_m$, i.e.,
		$\rho_{1:m}\triangleq
		\max_{p\in\mathcal I\cup\mathcal J}\rho_{1:m,p}$
		\tabularnewline
		
		$\rho_{1:m}^k$
		& The maximum cumulative load on core $k$ for coflows
		in $L_m$, i.e.,
		$\rho_{1:m}^k\triangleq
		\max_{p\in\mathcal I\cup\mathcal J}\rho_{1:m,p}^k$
		\tabularnewline
		
		$\tau_{1:m}$
		& The maximum cumulative number of nonzero flows incident
		to any port for coflows in $L_m$, i.e.,
		$\tau_{1:m}\triangleq
		\max_{p\in\mathcal I\cup\mathcal J}\tau_{1:m,p}$
		\tabularnewline
		
		$\tau_{1:m}^k$
		& The maximum cumulative number of nonzero subflows incident
		to any port on core $k$ for coflows in $L_m$, i.e.,
		$\tau_{1:m}^k\triangleq
		\max_{p\in\mathcal I\cup\mathcal J}\tau_{1:m,p}^k$
		\tabularnewline
		
		$\mathcal X_{1:m}^k$
		& The prefix workload state of core $k$, consisting of
		cumulative port loads and reconfiguration counts\tabularnewline
		
		$\varPhi^k(\mathcal X_{1:m}^k)$
		& The maximum cumulative port workload on core $k$
		for coflows in $L_m$\tabularnewline
		
		$b_\ell,G_\ell$
		& The upper threshold for LP completion times
		and the coflow index set of group $\ell$,
		respectively\tabularnewline
		
		$A_\ell,s_\ell,F_\ell$
		& The latest coflow release time, start time,
		and completion time of group $\ell$,
		respectively\tabularnewline
		
		$B_\ell$
		& The duration of group $\ell$, i.e.,
		$B_\ell=F_\ell-s_\ell$\tabularnewline
		
		$W_{\ell,p}^k,W_\ell^k$
		& The workload of group $\ell$ at port $p$ on core $k$
		and its maximum over ports, respectively\tabularnewline
		
		\bottomrule
	\end{tabular}
\end{table*}

\subsection{Network Architecture}
We consider a heterogeneous multi-core OCS network, as shown in Fig.
\ref{fig:HPN}. The network consists of $K$ independent OCS cores
operating in parallel, indexed by $\mathcal{K}=\left\{ 1,\ldots,K\right\} $.
Each core $k\in\mathcal{K}$ is an $N\times N$ non-blocking optical
switching fabric with per-port transmission rate $r^{k}$. Heterogeneity
arises from rate diversity across cores, i.e., $r^{k}$ varies with
$k$. When all cores have the same rate, $r^{k}=r$ for all $k\in\mathcal{K}$,
the network reduces to a homogeneous multi-core OCS architecture.

\begin{figure}[htp]
	\centering
	\includegraphics[width=0.90\columnwidth]{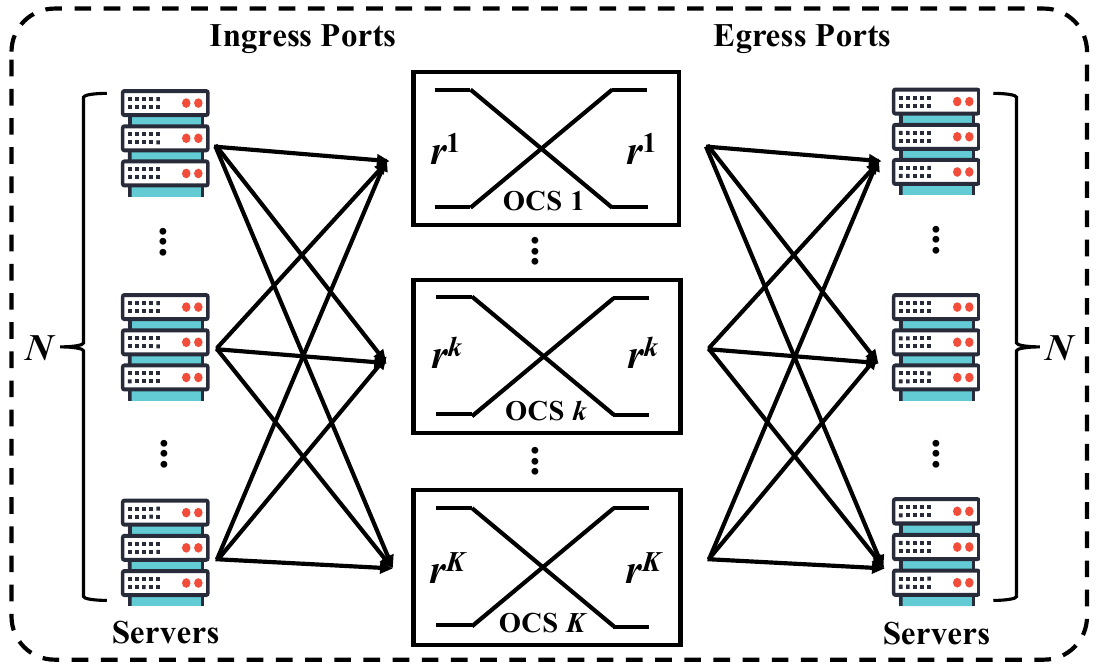}
	\caption{A heterogeneous multi-core OCS network.}
	\label{fig:HPN}
\end{figure}

The network has $N$ ingress ports and $N$ egress ports, indexed
by $\mathcal{I}=\left\{ 1,\ldots,N\right\} $ and $\mathcal{J}=\left\{ 1,\ldots,N\right\} $,
respectively. Ingress and egress ports are distinct resources even
when their numerical indices coincide. Each ingress and egress port is connected to all cores
through parallel uplinks and downlinks, so traffic from ingress port
$i\in\mathcal{I}$ to egress port $j\in\mathcal{J}$ can be assigned
to any core $k$. Within each core, feasible transmissions must form
a matching between ingress and egress ports;
equivalently, each ingress or egress port on that core
can participate in at most one active circuit at a time. Scheduling therefore requires jointly determining
the traffic allocation across cores and the circuit
schedules within each core.

\subsection{Traffic Abstraction}
We adopt the coflow abstraction \cite{networking} to model application-level
communication in multi-core OCS networks. A coflow represents a group
of semantically related flows generated by a communication stage,
and it completes only when all constituent flows finish.

Let $\mathcal{C}\triangleq\left\{ C_{m}:m\in\mathcal{M}\right\} $
denote the set of coflows, where $\mathcal{M}=\left\{ 1,\ldots,M\right\} $.
Each coflow $C_{m}$ is represented by an $N\times N$ demand matrix
$D_{m}=\left[d_{m}\left(i,j\right)\right]_{i\in\mathcal{I},j\in\mathcal{J}}$,
where $d_{m}\left(i,j\right)$ denotes the traffic demand from ingress
port $i$ to egress port $j$. If $d_{m}\left(i,j\right)>0$, the
corresponding flow is denoted by $f_{m}\left(i,j\right)$. The set
of nonzero flows in $C_{m}$ is $\mathcal{F}_{m}\triangleq\left\{ f_{m}\left(i,j\right)\mid d_{m}\left(i,j\right)>0,i\in\mathcal{I},j\in\mathcal{J}\right\} $.

\subsection{Problem Definition}
Consider a set of coflows $\mathcal{C}$, where each coflow $C_{m}\in\mathcal{C}$
is associated with a demand matrix $D_{m}=\left[d_{m}\left(i,j\right)\right]_{i\in\mathcal{I},j\in\mathcal{J}}$,
a positive weight $w_{m}$, and a release time $a_{m}\geq0$. Each coflow contains at least one flow with positive demand. The
goal is to schedule all nonzero flows in a $K$-core OCS network under
the \textit{not-all-stop} reconfiguration model so as to minimize
the total weighted coflow completion time (CCT). A feasible schedule
consists of the following main components:
\begin{itemize}
\item \textit{Inter-Core Flow Allocation:} For each coflow $C_{m}$ with
matrix $D_{m}$, determine an allocation $\left\{ D^{k}_{m}\right\} _{k\in\mathcal{K}}$
such that $D_{m}=\sum_{k\in\mathcal{K}}D^{k}_{m}$. Here, $D^{k}_{m}=\left[d^{k}_{m}\left(i,j\right)\right]_{i\in\mathcal{I},j\in\mathcal{J}}$
denotes the portion of $D_{m}$ allocated to core $k$, satisfying
$\sum_{k\in\mathcal{K}}d^{k}_{m}\left(i,j\right)=d_{m}\left(i,j\right)$,
$\forall\left(i,j\right)\in\mathcal{I}\times\mathcal{J}$.
\item \textit{Intra-Core Circuit Scheduling:} For each coflow
$C_m$ and core $k\in\mathcal{K}$, determine the circuit schedule
$S_m^k=\left\{(i,j,t_m^k(i,j))\mid
d_m^k(i,j)>0\right\}$,
where $t_m^k(i,j)$ denotes the time at which circuit reconfiguration
for subflow $f_m^k(i,j)$ begins. Each subflow can be scheduled
only after its corresponding coflow is released, i.e.,
$t_m^k(i,j)\ge a_m$ for every $d_m^k(i,j)>0$.
\end{itemize}
Under the \textit{not-all-stop} model, each circuit
reconfiguration incurs a delay $\delta$.
Each subflow incurs this delay separately, even when
subflows from different coflows share the same ingress
and egress ports.
If circuit reconfiguration for $f_m^k(i,j)$ begins at
time $t_m^k(i,j)$, transmission starts at
$t_m^k(i,j)+\delta$ and completes at
$T_m^k(i,j)=t_m^k(i,j)+\delta+d_m^k(i,j)/r^k$.
Both ports remain occupied throughout the interval
$[t_m^k(i,j),T_m^k(i,j))$, including reconfiguration
and transmission. For each core $k$ with $D_m^k\neq\mathbf{0}$,
the completion time of its assigned portion of coflow
$C_m$ is
$T_m^k=\max_{(i,j):\,d_m^k(i,j)>0}T_m^k(i,j)$.
The CCT of $C_m$ is
$T_m=\max_{k\in\mathcal K:\,D_m^k\neq\mathbf{0}}T_m^k$.
The objective is to minimize the total weighted CCT:
$\min\sum_{m\in\mathcal M}w_mT_m$.

\subsection{Hardness Analysis}
The problem is computationally challenging even in the single-core,
single-coflow setting. When the reconfiguration delay is ignored,
i.e., $\delta=0$, scheduling a single coflow in a single-core OCS
network is equivalent to the nonpreemptive (i.e., once a flow starts its reconfiguration, it proceeds without
interruption until transmission completes) open-shop scheduling problem
for minimizing the makespan, which is known to be NP-hard \cite{gonzalez1976open,sunflow}.
Therefore, single-coflow scheduling in a single-core OCS network is
NP-hard.

The multi-core OCS scheduling problem studied in this paper generalizes
this hard special case. Specifically, by setting $K=1$, $M=1$, and
$\delta=0$, the problem reduces to the single-core, single-coflow
case above. Hence, the multi-coflow scheduling problem in multi-core
OCS networks is NP-hard.

\section{Multi-Core Coflow Scheduling\label{sec:Multiple-Coflow-Scheduling}}
This section presents an LP-guided scheduling framework for multi-coflow
scheduling in heterogeneous multi-core OCS networks. The LP-guided global order is integrated with
the subsequent inter-core flow allocation and intra-core circuit scheduling
stages. We describe the complete approximation algorithm in Algorithm
\ref{alg:alg1} and establish its provable performance guarantees.

\subsection{Algorithm Framework}
The scheduling framework consists of three components:
(i) global coflow ordering based on LP,
(ii) inter-core flow allocation, and
(iii) intra-core circuit scheduling.
The global order determines the sequence in which coflows are
processed during allocation and provides priorities for
subsequent circuit scheduling.

\subsubsection{LP-guided Coflow Ordering}
Prior work \cite{wang2026scheduling} established the lower bound
$T_{\textrm{LB}}(D_m)\triangleq\delta+\rho_m/R$
for a single coflow, independently of its allocation,
where $\rho_{m}$
is the maximum port load in $D_{m}$, and $R\triangleq\sum_{k\in\mathcal{K}}r^{k}$
denotes the aggregated port rate across all $K$ OCS cores. However, this bound applies only to an individual coflow and does not capture prefix interactions in multi-coflow scheduling, which are crucial for analyzing the total weighted CCT. The analysis in prior work based on this bound yields an approximation factor that scales with the number of coflows $M$ \cite{wang2026scheduling}.

To overcome this limitation, we formulate an integer linear
programming (ILP) relaxation using coflow completion order
and capacity constraints,
and then relax its binary variables to obtain a linear programming (LP) relaxation,
which provides a global lower bound by capturing the
cumulative transmission and reconfiguration workloads induced by preceding
coflows. Consider a coflow $C_{m}$ with demand matrix $D_{m}=\big[d_{m}\left(i,j\right)\big]_{i\in\mathcal{I},j\in\mathcal{J}}$.
For any port $p\in\mathcal{I}\cup\mathcal{J}$, define the traffic
load incident to port $p$ as
\[
\rho_{m,p}\triangleq\begin{cases}
\sum_{j\in\mathcal{J}}d_{m}(i,j), & \text{if }p=i\in\mathcal{I},\\
\sum_{i\in\mathcal{I}}d_{m}(i,j), & \text{if }p=j\in\mathcal{J},
\end{cases}
\]
and the number of nonzero flow entries incident to port $p$ as
\[
\tau_{m,p}\triangleq\begin{cases}
\sum_{j\in\mathcal{J}}\mathbf{1}[d_{m}(i,j)>0], & \text{if }p=i\in\mathcal{I},\\
\sum_{i\in\mathcal{I}}\mathbf{1}[d_{m}(i,j)>0], & \text{if }p=j\in\mathcal{J}.
\end{cases}
\]

We introduce a binary ordering variable $x_{m,m'}$ for each
pair of distinct coflows $C_m$ and $C_{m'}$. A feasible schedule
induces an order according to coflow completion times, with
ties broken consistently. In this order, $x_{m,m'}=1$ if
$C_m$ precedes $C_{m'}$, and $x_{m,m'}=0$ otherwise.
Using the transmission and reconfiguration capacity constraints
at each port, we formulate the following ILP relaxation:

\begin{align}
\text{(ILP)}\quad\min & \sum_{m\in\mathcal{M}}w_{m}T_{m}\tag{1a}\label{eq:ip-obj}\\
\text{s.t.}\quad & T_{m}\ge\frac{1}{R}\left(\rho_{m,p}+\sum_{m'\neq m}\rho_{m',p}x_{m',m}\right) &  & \forall m,p\tag{1b}\label{eq:ip-trans}\\
 & T_{m}\ge\frac{\delta}{K}\left(\tau_{m,p}+\sum_{m'\neq m}\tau_{m',p}x_{m',m}\right) &  & \forall m,p\tag{1c}\label{eq:ip-reconfig}\\
 & T_{m}\ge a_{m} &  & \forall m\tag{1d}\label{eq:ip-release}\\
 & x_{m,m'}+x_{m',m}=1 &  & \forall m\neq m'\tag{1e}\label{eq:ip-order}\\
 & x_{m,m'}\in\left\{ 0,1\right\}  &  & \forall m\neq m'\tag{1f}\label{eq:ip-binary}
\end{align}

The objective function \eqref{eq:ip-obj} minimizes the total weighted
CCT in the formulation. Constraint \eqref{eq:ip-trans}
is the transmission capacity constraint. For each coflow $C_{m}$
and port $p$, by the time $C_{m}$ completes, all coflows ordered
before $C_{m}$, together with $C_{m}$ itself, must have already
transmitted all traffic incident to port $p$. Since the aggregated
transmission rate available on port $p$ across all OCS cores is $R$,
the cumulative transmitted load by time $T_{m}$ cannot exceed $RT_{m}$.
Constraint \eqref{eq:ip-reconfig} is the reconfiguration capacity
constraint. Similarly, when $C_{m}$ completes, all reconfiguration
operations associated with coflows ordered before $C_{m}$, together
with those required by $C_{m}$, must also have been completed on
port $p$. For $\delta>0$, since the aggregated reconfiguration processing rate available
on port $p$ is upper bounded by $K/\delta$, the cumulative number
of such reconfiguration operations cannot exceed $\left(K/\delta\right)T_{m}$.
Constraint \eqref{eq:ip-release} enforces release times, since a
coflow cannot complete before it is released. Constraint \eqref{eq:ip-order}
ensures that exactly one ordering relation is selected for each pair
of coflows. Constraint \eqref{eq:ip-binary} imposes integrality on
the ordering variables.

To obtain the LP relaxation, we relax Constraint \eqref{eq:ip-binary}
and allow each ordering variable to take a fractional value in $[0,1]$.
Hence, the LP formulation is 
\begin{align}
	\text{(LP)}\quad\min \quad
	& \sum_{m\in\mathcal{M}}w_mT_m
	\tag{2a}\label{eq:lp-obj}\\
	\text{s.t.}\quad
	& \text{Constraints }\eqref{eq:ip-trans}\text{--}\eqref{eq:ip-order}
	\nonumber\\
	& 0\le x_{m,m'}\le1
	&& \forall m\neq m'
	\tag{2b}\label{eq:lp-relax}
\end{align}

Let $\bigl(\widetilde{T}_{m},\widetilde{x}_{m,m'}\bigr)$ denote an
optimal solution to the LP relaxation, and let $T^{*}_{m}$ denote
the completion time of $C_{m}$ in an optimal feasible schedule for
the original problem. Since any feasible schedule induces a feasible
integral ordering solution to the ILP, the LP formulation is a valid
relaxation of the original scheduling problem. Consequently, the objective
value of the LP relaxation provides a valid lower bound on the optimal
value of the original problem, i.e.,
\[
\sum^{M}_{m=1}w_{m}\widetilde{T}_{m}\le\sum^{M}_{m=1}w_{m}T^{*}_{m}.
\]

Finally, we extract a global coflow priority order from the LP completion times $\left\{ \widetilde{T}_{m}\right\} _{m\in\mathcal{M}}$. Specifically,
let $\pi$ be an ordering of coflows such that $\widetilde{T}_{\pi\left(1\right)}\le\widetilde{T}_{\pi\left(2\right)}\le\cdots\le\widetilde{T}_{\pi\left(M\right)}$,
with ties broken arbitrarily. For notational simplicity, we reindex
the coflows according to this LP-guided order, so that $\widetilde{T}_{1}\le\widetilde{T}_{2}\le\cdots\le\widetilde{T}_{M}$.
Throughout the rest of the paper, unless otherwise specified, the
index $m$ refers to the $m$-th coflow in this reindexed LP-guided
order. Thus, coflows with smaller LP completion times are assigned
higher priorities. For any $m\in\mathcal{M}$, we define $L_{m}\triangleq\left\{ 1,\ldots,m\right\} $
as the prefix consisting of the first $m$ coflows in the reindexed
LP-guided order. 

\subsubsection{Inter-Core Flow Allocation}
The inter-core flow allocation is performed by a prefix-aware greedy
procedure. Coflows are processed sequentially according to the reindexed
LP-guided order. In our allocation, each flow is assigned entirely
to a single core; flow splitting is not allowed to avoid packet reordering,
buffering overhead, and additional control plane complexity in practical
multi-core OCS deployments \cite{huang2020weaver}. 

To guide allocation, we define a prefix workload measure for each
core. For the first $m$ coflows and core $k$, let $\rho^{k}_{1:m,p}\triangleq\sum^{m}_{\ell=1}\rho^{k}_{\ell,p}$
denote the cumulative traffic load incident to port $p$ on core $k$,
and let $\tau^{k}_{1:m,p}\triangleq\sum^{m}_{\ell=1}\tau^{k}_{\ell,p}$
denote the cumulative number of nonzero subflows incident to port
$p$ on core $k$. The prefix workload state of core $k$ is 
\[
\mathcal{X}^{k}_{1:m}\triangleq\left(\left\{ \rho^{k}_{1:m,p}\right\} _{p\in\mathcal{I}\cup\mathcal{J}},\left\{ \tau^{k}_{1:m,p}\right\} _{p\in\mathcal{I}\cup\mathcal{J}}\right).
\]

The prefix workload measure of core $k$ is defined as 
\[
\varPhi^{k}\left(\mathcal{X}^{k}_{1:m}\right)\triangleq\max_{p\in\mathcal{I}\cup\mathcal{J}}\left(\frac{\rho^{k}_{1:m,p}}{r^{k}}+\tau^{k}_{1:m,p}\delta\right),
\]
which captures the cumulative transmission and reconfiguration workloads
on the bottleneck port of core $k$.

During the allocation of coflow $C_{m}$, the temporary state of each
core $\mathcal{X}^{k}_{1:m}$ is initialized as $\mathcal{X}^{k}_{1:m-1}$
and is updated as flows of $C_{m}$ are assigned. For each flow $f_{m}\left(i,j\right)$,
the algorithm tentatively evaluates the resulting prefix workload
measure on each core and assigns the entire flow to the core that
yields the smallest value:
\[
k^{*}\leftarrow\arg\min_{k\in\mathcal{K}}\varPhi^{k}\left(\mathcal{X}^{k}_{1:m}\oplus f_{m}\left(i,j\right)\right),
\]
where $\mathcal{X}^{k}_{1:m}\oplus f_{m}\left(i,j\right)$ denotes
the state obtained by tentatively assigning $f_{m}\left(i,j\right)$
to core $k$. Specifically, the operator $\oplus$ increases the cumulative
traffic loads of ingress port $i$ and egress port $j$
by $d_m(i,j)$ and increases the reconfiguration count at each
of these ports by one, while leaving all other ports unchanged.

Intuitively, this rule seeks to balance cumulative transmission
and reconfiguration workloads across cores, thereby mitigating
queuing and blocking in subsequent circuit scheduling.
Although the processing order of flows within a coflow does
not affect the approximation guarantee, allocating larger
flows first may reduce the total weighted CCT in practice
\cite{huang2020weaver,wang2026scheduling}.

\subsubsection{Intra-Core Circuit Scheduling}
After inter-core flow allocation, coflows are partitioned into
geometric groups according to their LP completion times.
Let $b_0=\min_{m\in\mathcal M}\widetilde T_m$ and
$b_\ell=2^\ell b_0$ for $\ell\ge0$.
Since each coflow contains at least one flow with positive demand,
the transmission constraints imply $b_0>0$.
Define
\[
\begin{aligned}
	G_0&=\{m:0<\widetilde T_m\le b_0\},\\
	G_\ell&=\{m:b_{\ell-1}<\widetilde T_m\le b_\ell\},
	\quad \ell\ge1.
\end{aligned}
\]
Only nonempty groups are processed, in increasing order of $\ell$.
All groups retain the flow assignments determined by the allocation
procedure, which accumulates workloads over the entire LP order.
Each group starts only after all its coflows have been released
and the preceding group has completed on every core.
Within a group, all cores process their assigned flows in parallel.

Each core applies a nonpreemptive greedy scheduler guided by
the LP coflow order.
A flow starts only when both its ingress and egress ports
are idle and occupies them for $\delta+d_m^k(i,j)/r^k$,
including reconfiguration and transmission.


\subsection{Approximation Algorithm}
This subsection describes the complete approximation algorithm, as shown in Algorithm~\ref{alg:alg1}, which determines the coflow priority order, allocates flows across cores, and invokes Algorithm 2 to construct the circuit schedules.

\textbf{LP-guided coflow ordering.}
Algorithm~\ref{alg:alg1} first solves the LP relaxation and sorts coflows
in nondecreasing order of their LP completion times
$\widetilde T_m$ (Lines 1--2).
Coflows are reindexed according to this order, which
guides both flow allocation and circuit scheduling.

\begin{algorithm}[h]
\caption{Multi-Coflow Scheduling in Multi-Core OCS Networks}
\label{alg:alg1} \textbf{Input:} demand matrices $\left\{ D_{m}=\left[d_{m}\left(i,j\right)\right]\right\} _{m\in\mathcal{M}}$;
weights $\left\{ w_{m}\right\} _{m\in\mathcal{M}}$; release times
$\left\{ a_{m}\right\} _{m\in\mathcal{M}}$; core rates $\left\{ r^{k}\right\} _{k\in\mathcal{K}}$;
reconfiguration delay $\delta$\\
 \textbf{Output:} flow allocations $\left\{ D^{k}_{m}\right\} _{m\in\mathcal{M},k\in\mathcal{K}}$
and circuit schedules $\left\{ S^{k}_{m}\right\} _{m\in\mathcal{M},k\in\mathcal{K}}$
for all cores \begin{algorithmic}[1]\Linecomment COFLOW ORDERING
\State Solve the LP relaxation and obtain $\left\{ \widetilde{T}_{m}\right\} _{m\in\mathcal{M}}$
\State Sort coflows in nondecreasing order of $\widetilde{T}_{m}$
and reindex: $\widetilde{T}_{1}\le\widetilde{T}_{2}\le\cdots\le\widetilde{T}_{M}$
\Linecomment FLOW ALLOCATION\State Initialize $\mathcal{X}^{k}_{1:0}\leftarrow\mathbf{0}$,
$\forall\ensuremath{k\in\mathcal{K}}$\For{$m=1$ to $M$}\State
Initialize $\ensuremath{\mathcal{X}^{k}_{1:m}\leftarrow\mathcal{X}^{k}_{1:m-1}}$,
$D^{k}_{m}\leftarrow\mathbf{0}_{N\times N}$, $\forall k\in\mathcal{K}$
\State $\mathcal{F}_{m}\triangleq\left\{ f_{m}\left(i,j\right)\mid d_{m}\left(i,j\right)>0\right\} $
\State Sort $\mathcal{F}_{m}$ in nonincreasing order of $d_{m}\left(i,j\right)$
\For{each flow $f_{m}\left(i,j\right)\in\mathcal{F}_{m}$}\State
$k^{*}\leftarrow\mathrm{argmin}_{k\in\mathcal{K}}\varPhi^{k}\left(\mathcal{X}^{k}_{1:m}\oplus f_{m}\left(i,j\right)\right)$\State
Allocate the entire flow $f_{m}\left(i,j\right)$ to core $k^{*}$:
$f^{k^{*}}_{m}\left(i,j\right)\leftarrow f_{m}\left(i,j\right)$ \State
$D^{k^{*}}_{m}\leftarrow D^{k^{*}}_{m}+d_{m}\left(i,j\right)E^{ij}$
\State $\mathcal{X}^{k^{*}}_{1:m}\leftarrow\mathcal{X}^{k^{*}}_{1:m}\oplus f_m(i,j)$
\EndFor\EndFor\Linecomment CIRCUIT SCHEDULING
\State Obtain circuit schedules
$\{S_m^k\}_{m\in\mathcal{M},k\in\mathcal{K}}$
using Algorithm~\ref{alg:grouped-scheduling}
\end{algorithmic}
\end{algorithm}

\textbf{Inter-core flow allocation.}
 Subsequently, Algorithm~\ref{alg:alg1} enters the flow allocation phase (Lines
3--14). For each core $k\in\mathcal K$, Algorithm~\ref{alg:alg1} maintains
a prefix workload state $\mathcal X_{1:m}^k$ that records
the cumulative traffic load and number of nonzero subflows
at each ingress and egress port for the first $m$ coflows. Initially, $\mathcal{X}^{k}_{1:0}$
is set to zero for every core $k$ (Line 3). For each coflow $C_{m}$
processed sequentially (Line 4), the workload state $\mathcal{X}^{k}_{1:m}$
is initialized from $\mathcal{X}^{k}_{1:m-1}$, and the matrix $D^{k}_{m}$
is initialized to the zero matrix for every core $k$ (Line 5). The
set of nonzero flows in $C_{m}$ is then defined as $\mathcal{F}_{m}=\left\{ f_{m}\left(i,j\right)\mid d_{m}\left(i,j\right)>0\right\} $
(Line 6), and these flows are sorted in nonincreasing order of their
sizes $d_{m}\left(i,j\right)$ (Line 7). For each flow $f_{m}\left(i,j\right)\in\mathcal{F}_{m}$
(Line 8), Algorithm~\ref{alg:alg1} tentatively assigns it to each core and evaluates
the resulting prefix workload measure $\varPhi^{k}\left(\mathcal{X}^{k}_{1:m}\oplus f_{m}\left(i,j\right)\right)$. The selected core is $k^{*}=\arg\min_{k\in\mathcal{K}}\varPhi^{k}\left(\mathcal{X}^{k}_{1:m}\oplus f_{m}\left(i,j\right)\right)$
(Line 9). The entire flow $f_{m}\left(i,j\right)$ is then
allocated to core $k^{*}$, i.e., $f^{k^{*}}_{m}\left(i,j\right)\leftarrow f_{m}\left(i,j\right)$,
equivalently $d^{k^{*}}_{m}\left(i,j\right)=d_{m}\left(i,j\right)$
(Line 10). Accordingly, the allocation matrix $D^{k^{*}}_{m}$ is
updated by adding the demand $d_{m}\left(i,j\right)$ to entry $\left(i,j\right)$,
i.e., $D^{k^{*}}_{m}\leftarrow D^{k^{*}}_{m}+d_{m}\left(i,j\right)E^{ij}$,
where $E^{ij}$ is the matrix with a single one at entry $\left(i,j\right)$
and zeros elsewhere (Line 11). Finally, $\mathcal{X}^{k^{*}}_{1:m}$
is updated to include this flow (Line 12). After all flows in $C_{m}$
have been allocated, the matrices $\left\{ D^{k}_{m}\right\} _{k\in\mathcal{K}}$
represent the inter-core allocation of coflow $C_{m}$.

\begin{algorithm}[t]
	\caption{Intra-Core Circuit Scheduling with LP Grouping}
	\label{alg:grouped-scheduling}
	\textbf{Input:}
	flow allocations
	$\{D_m^k\}_{m\in\mathcal M,k\in\mathcal K}$;
	LP completion times $\{\widetilde T_m\}_{m\in\mathcal M}$;
	release times $\{a_m\}_{m\in\mathcal M}$;
	core rates $\{r^k\}_{k\in\mathcal K}$;
	reconfiguration delay $\delta$\\
	\textbf{Output:}
	circuit schedules $\{S_m^k\}_{m\in\mathcal M,k\in\mathcal K}$
	\begin{algorithmic}[1]
		\State Form groups $G_\ell$ from the LP
		completion times
		\State $F_{\mathrm{prev}}\gets0$;
		initialize $S_m^k\gets\emptyset$ for all $m,k$
		\For{each nonempty $G_\ell$ in increasing order of $\ell$}
		\State $A_\ell\gets\max_{m\in G_\ell}a_m$;
		$s_\ell\gets\max\{F_{\mathrm{prev}},A_\ell\}$
		\ForAll{cores $k\in\mathcal K$ in parallel}
		\State Initialize $t\gets s_\ell$;
		$\mathcal A^k\gets\emptyset$
		\State $\mathcal U^k\gets
		\{f_m^k(i,j):m\in G_\ell,\ d_m^k(i,j)>0\}$
		\While{$\mathcal U^k\ne\emptyset$
			or $\mathcal A^k\ne\emptyset$}
		\State Release the ingress and egress ports
		of all flows in $\mathcal A^k$ completing at $t$;
		remove these flows from $\mathcal A^k$
		\For{each $f_m^k(i,j)\in\mathcal U^k$
			in LP coflow order}
		\If{ingress port $i$ and egress port $j$
			are both idle}
		\State $t_m^k(i,j)\gets t$
		\State $T_m^k(i,j)\gets
		t+\delta+d_m^k(i,j)/r^k$
		\State Mark ingress port $i$ and egress port $j$
		as busy until $T_m^k(i,j)$
		\State Add $(i,j,t)$ to $S_m^k$
		\State Move $f_m^k(i,j)$ from
		$\mathcal U^k$ to $\mathcal A^k$
		\EndIf
		\EndFor
		\If{$\mathcal A^k\ne\emptyset$}
		\State $t\gets
		\min_{f_m^k(i,j)\in\mathcal A^k}T_m^k(i,j)$
		\EndIf
		\EndWhile
		\State $F_\ell^k\gets t$
		\EndFor
		\State $F_\ell\gets\max_{k\in\mathcal K}F_\ell^k$;
		$F_{\mathrm{prev}}\gets F_\ell$
		\EndFor
	\end{algorithmic}
\end{algorithm}

\textbf{Intra-core circuit scheduling.} Algorithm~\ref{alg:alg1} invokes
Algorithm~\ref{alg:grouped-scheduling}
to schedule the allocated flows (Line 15). Algorithm~\ref{alg:grouped-scheduling}
constructs circuit schedules by grouping coflows according
to their LP completion times.
It initializes the schedules and sets $F_{\mathrm{prev}}=0$,
where $F_{\mathrm{prev}}$ records the completion time
of the preceding nonempty group (Lines 1--2).
For each nonempty group $G_\ell$,
$A_\ell=\max_{m\in G_\ell}a_m$ is the latest release time
among its coflows.
The group starts at
$s_\ell=\max\{F_{\mathrm{prev}},A_\ell\}$,
after all its coflows have been released and the preceding
group has completed on every core (Lines 3--4). All cores process their assigned flows in parallel. Each core initializes $t=s_\ell$,
an empty active flow set $\mathcal A^k$,
and an unstarted flow set $\mathcal U^k$ containing
its assigned flows in $G_\ell$ (Lines 5--7).
The event loop runs until both sets are empty
(Lines 8--22).
At each event, the scheduler releases the ingress
and egress ports of all flows completing at $t$
and removes them from $\mathcal A^k$ (Line 9).
It then scans all unstarted flows in LP coflow order,
using a fixed order within each coflow (Line 10).
If both ports of a flow are idle, the scheduler records
its reconfiguration start time (Line 12) and completion time (Line 13)
and marks both ports as busy until completion (Line 14).
It adds $(i,j,t)$ to $S_m^k$ (Line 15) and moves the flow
from $\mathcal U^k$ to $\mathcal A^k$ (Line 16). If active flows remain, $t$ advances to their earliest
completion time (Lines 19--20). Once a core has completed its assigned flows in the group,
it records its completion time $F_\ell^k=t$ (Line 23).
The group completes at
$F_\ell=\max_{k\in\mathcal K}F_\ell^k$,
when all cores have finished processing their assigned flows.
The algorithm updates $F_{\mathrm{prev}}=F_\ell$
before processing the next nonempty group (Line 25).

\textit{Time complexity:}
The LP relaxation has $O(M^2)$ variables and
$O(MN+M^2)$ constraints and can be solved in polynomial time.
Excluding the LP solution time, Algorithm~\ref{alg:alg1},
including Algorithm~\ref{alg:grouped-scheduling},
runs in $O(MN^3K+M^2N^4)$ time.
At most $MN^2$ flows are allocated, and evaluating all cores
for each flow takes $O(NK)$ time by scanning their ports.
Circuit scheduling processes at most $MN^2$ completion events
across all cores, with each event scanning at most $MN^2$ flows.

\subsection{Analysis of Performance Guarantees
	\label{sec:Theoretical-Analysis_Multiple}}
For each coflow $C_m$ and each port
$p\in\mathcal I\cup\mathcal J$,
let $\rho_{m,p}$ denote the traffic load incident to port $p$,
and let $\tau_{m,p}$ denote the number of nonzero flows
incident to port $p$.
For any prefix $L_m=\{1,\ldots,m\}$ in the LP-guided order,
define $\rho_{1:m,p}\triangleq\sum_{\ell=1}^m\rho_{\ell,p}$
and $\tau_{1:m,p}\triangleq\sum_{\ell=1}^m\tau_{\ell,p}$.
Let $\rho_{1:m}\triangleq
\max_{p\in\mathcal I\cup\mathcal J}\rho_{1:m,p}$
and $\tau_{1:m}\triangleq
\max_{p\in\mathcal I\cup\mathcal J}\tau_{1:m,p}$.
Let $R=\sum_{k\in\mathcal K}r^k$ denote the aggregated
port transmission rate and
$r_{\max}\triangleq\max_{k\in\mathcal K}r^k$. 

For each nonempty group $G_\ell$, let
$B_\ell=F_\ell-s_\ell$ denote its duration,
where $s_\ell$ and $F_\ell$ are its start and
completion times, respectively.
Let $W_{\ell,p}^k\triangleq
\sum_{m\in G_\ell}
(\rho_{m,p}^k/r^k+\delta\tau_{m,p}^k)$
denote the total transmission and reconfiguration
workload assigned to port $p$ on core $k$ within the group,
and let $W_\ell^k\triangleq
\max_{p\in\mathcal I\cup\mathcal J}W_{\ell,p}^k$
denote the maximum port workload on that core.

The performance analysis follows the three components
of Algorithm~\ref{alg:alg1}.
\begin{enumerate}
	\item LP-guided ordering:
	The ordering policy sorts coflows in nondecreasing order of their LP completion times
	$\{\widetilde T_m\}_{m\in\mathcal M}$, rather than directly
	rounding the fractional ordering variables $\widetilde x_{m,m'}$. To analyze allocation and grouped scheduling, we bound
	each prefix workload by its largest LP completion time:
	$\rho_{1:m}\le2R\widetilde T_m$ and
	$\delta\tau_{1:m}\le2K\widetilde T_m$
	(Lemmas~\ref{transmission-capacity-prefix-bound}
	and~\ref{reconfiguration-capacity-prefix-bound}, respectively).

	\item Inter-core flow allocation: The greedy allocation bounds the maximum cumulative port workload on each core in terms of the global prefix workloads.
	Specifically,
	$\max_{k\in\mathcal K}\varPhi^k(\mathcal X_{1:m}^k)
	\le\rho_{1:m}/r_{\max}+\delta\tau_{1:m}$
	(Lemma~\ref{allocation-phase-prefix-bound}).
	
	\item Intra-core circuit scheduling:
	The circuit scheduling scheme ensures that each group's
	duration is at most twice its maximum per-core port workload:
	$B_\ell\le2\max_{k\in\mathcal K}W_\ell^k$
	(Lemma~\ref{scheduling-phase-group-bound}).
	For $q_\ell=\max\{m:\widetilde T_m\le b_\ell\}$,
	the inclusion $G_\ell\subseteq L_{q_\ell}$ gives
	$\max_{k\in\mathcal K}W_\ell^k
	\le\max_{k\in\mathcal K}
	\varPhi^k(\mathcal X_{1:q_\ell}^k)$.
	Applying Lemma~\ref{allocation-phase-prefix-bound}
	with $m=q_\ell$ therefore yields
	$B_\ell\le
	2(\rho_{1:q_\ell}/r_{\max}+\tau_{1:q_\ell}\delta)$.
	We then use the LP prefix bounds to control group durations
	and account for waiting for coflow releases to derive
	a completion time bound for each coflow.
\end{enumerate}

Combining these bounds with the LP lower bound yields
the approximation guarantees for multi-core OCS networks.
Applying the same analysis with the reconfiguration term
removed yields the guarantees for multi-core EPS networks.
These results are established in
Theorems~\ref{theorem1} and~\ref{theorem2}
and Corollaries~\ref{cor:ocs-zero-release}
and~\ref{cor:eps-zero-release}.

\subsubsection{Prefix Bounds for LP Ordering}
\begin{lemma}[Transmission Capacity Prefix Bound] For any prefix
set $L_{m}=\left\{ 1,\ldots,m\right\} $ and any fixed port $p\in\mathcal{I}\cup\mathcal{J}$,
$\sum^{m}_{\ell=1}\rho_{\ell,p}\widetilde{T}_{\ell}\ge\frac{1}{2R}\left(\sum^{m}_{\ell=1}\rho_{\ell,p}\right)^{2}$.
Consequently, $\rho_{1:m}\le2R\widetilde{T}_{m}$.\label{transmission-capacity-prefix-bound}\end{lemma}
\begin{IEEEproof}
By the transmission capacity constraints in Eq. \eqref{eq:ip-trans},
for each $\ell\in L_{m}$,
\[
\widetilde{T}_{\ell}\ge\frac{1}{R}\left(\rho_{\ell,p}+\sum_{q\neq\ell}\rho_{q,p}\widetilde{x}_{q,\ell}\right).
\]

Since all terms are nonnegative, we may restrict the summation to
$q\in L_{m}\setminus\left\{ \ell\right\} $. Multiplying both sides
by $\rho_{\ell,p}$ and summing over $\ell=1,\ldots,m$, we obtain
\[
\sum^{m}_{\ell=1}\rho_{\ell,p}\widetilde{T}_{\ell}\ge\frac{1}{R}\left(\sum^{m}_{\ell=1}\rho^{2}_{\ell,p}+\sum^{m}_{\ell=1}\sum^{m}_{\substack{q=1\\
q\neq\ell
}
}\rho_{\ell,p}\rho_{q,p}\widetilde{x}_{q,\ell}\right).
\]

For any pair $1\le\ell<q\le m$, the LP ordering constraint gives
$\widetilde{x}_{\ell,q}+\widetilde{x}_{q,\ell}=1$. Therefore, 
\[
\rho_{\ell,p}\rho_{q,p}\widetilde{x}_{q,\ell}+\rho_{q,p}\rho_{\ell,p}\widetilde{x}_{\ell,q}=\rho_{\ell,p}\rho_{q,p}.
\]

Summing over all unordered pairs in $L_{m}$, we have
\[
\sum^{m}_{\ell=1}\sum^{m}_{\substack{q=1\\
q\neq\ell
}
}\rho_{\ell,p}\rho_{q,p}\widetilde{x}_{q,\ell}=\sum_{1\le\ell<q\le m}\rho_{\ell,p}\rho_{q,p}.
\]

Hence, 
\[
\begin{aligned}\sum^{m}_{\ell=1}\rho_{\ell,p}\widetilde{T}_{\ell} & \ge\frac{1}{R}\left(\sum^{m}_{\ell=1}\rho^{2}_{\ell,p}+\sum_{1\le\ell<q\le m}\rho_{\ell,p}\rho_{q,p}\right)\\
 & \ge\frac{1}{2R}\left(\sum^{m}_{\ell=1}\rho_{\ell,p}\right)^{2}.
\end{aligned}
\]

Since the coflows are reindexed in nondecreasing order of LP completion
times, we have $\widetilde{T}_{\ell}\le\widetilde{T}_{m}$, $\forall\ell\le m$.
Thus,
\[
\widetilde{T}_{m}\sum^{m}_{\ell=1}\rho_{\ell,p}\ge\frac{1}{2R}\left(\sum^{m}_{\ell=1}\rho_{\ell,p}\right)^{2},
\]
which implies 
\[
\sum^{m}_{\ell=1}\rho_{\ell,p}\le2R\widetilde{T}_{m}.
\]

Since $\rho_{1:m,p}=\sum^{m}_{\ell=1}\rho_{\ell,p}$, we have $\rho_{1:m,p}\le2R\widetilde{T}_{m}$.
Taking the maximum over all $p\in\mathcal{I}\cup\mathcal{J}$ gives
\[
\rho_{1:m}\le2R\widetilde{T}_{m}.
\]

This completes the proof.
\end{IEEEproof}

\begin{lemma}[Reconfiguration Capacity Prefix Bound] For $\delta>0$ and any prefix
set $L_{m}=\left\{ 1,\ldots,m\right\} $ and any fixed port $p\in\mathcal{I}\cup\mathcal{J}$,
$\sum^{m}_{\ell=1}\tau_{\ell,p}\widetilde{T}_{\ell}\ge\frac{\delta}{2K}\left(\sum^{m}_{\ell=1}\tau_{\ell,p}\right)^{2}$.
Consequently, $\tau_{1:m}\le\frac{2K}{\delta}\widetilde{T}_{m}$.\label{reconfiguration-capacity-prefix-bound}\end{lemma}
\begin{IEEEproof}
The proof follows the same pairwise ordering argument as in Lemma
\ref{transmission-capacity-prefix-bound}. By the reconfiguration capacity
constraints in Eq. \eqref{eq:ip-reconfig}, for each $\ell\in L_{m}$,
\[
\widetilde{T}_{\ell}\ge\frac{\delta}{K}\left(\tau_{\ell,p}+\sum_{q\neq\ell}\tau_{q,p}\widetilde{x}_{q,\ell}\right).
\]

Repeating the same argument as in the transmission capacity bound
gives 
\[
\sum^{m}_{\ell=1}\tau_{\ell,p}\le\frac{2K}{\delta}\widetilde{T}_{m}.
\]

Since $\tau_{1:m,p}=\sum^{m}_{\ell=1}\tau_{\ell,p}$, we have $\tau_{1:m,p}\le\frac{2K}{\delta}\widetilde{T}_{m}$.
Taking the maximum over all $p\in\mathcal{I}\cup\mathcal{J}$ gives
\[
\tau_{1:m}\le\frac{2K}{\delta}\widetilde{T}_{m}.
\]

This completes the proof.
\end{IEEEproof}

\subsubsection{Prefix Bound for Flow Allocation}

\begin{lemma}[Allocation Phase Prefix Bound] For any $m\in\mathcal{M}$,
the prefix workload states $\left\{ \mathcal{X}^{k}_{1:m}\right\} _{k\in\mathcal{K}}$
generated by the allocation phase satisfy $\max_{k\in\mathcal{K}}\varPhi^{k}\left(\mathcal{X}^{k}_{1:m}\right)\le\frac{\rho_{1:m}}{r_{\max}}+\tau_{1:m}\delta$.\label{allocation-phase-prefix-bound}\end{lemma}
\begin{IEEEproof}
Consider any nonempty core $k_{1}$ after the first $m$ coflows have
been processed. Let $\bar{f}^{k_{1}}\left(i,j\right)$ be the last
flow allocated to core $k_{1}$ among these coflows, and let $\mathcal{\bar{X}}^{k_{1}}$
be the workload state of core $k_{1}$ immediately before allocating
this flow. Then, $\mathcal{X}^{k_{1}}_{1:m}=\mathcal{\bar{X}}^{k_{1}}\oplus\bar{f}^{k_{1}}\left(i,j\right)$.

When $\bar{f}^{k_{1}}\left(i,j\right)$ is allocated, the greedy rule
selects the core with the minimum resulting prefix workload. Hence,
for any core $k_{2}\in\mathcal{K}$, 
\[
\varPhi^{k_{1}}\left(\mathcal{\bar{X}}^{k_{1}}\oplus\bar{f}^{k_{1}}\left(i,j\right)\right)\le\varPhi^{k_{2}}\left(\mathcal{\bar{X}}^{k_{2}}\oplus\bar{f}^{k_{1}}\left(i,j\right)\right),
\]
where $\bar{\mathcal X}^{k_2}$ is the workload state
of core $k_2$ immediately before allocating this flow.

For any core $k_{2}$, the tentative workload state $\mathcal{\bar{X}}^{k_{2}}\oplus\bar{f}^{k_{1}}\left(i,j\right)$
contains only workload contributed by flows among the first $m$ coflows.
Hence, for every port $p\in\mathcal{I}\cup\mathcal{J}$, its cumulative
traffic load is at most $\rho_{1:m,p}$, and its cumulative reconfiguration
count is at most $\tau_{1:m,p}$. Therefore, 
\[
\varPhi^{k_{2}}\left(\bar{\mathcal{X}}^{k_{2}}\oplus\bar{f}^{k_{1}}\left(i,j\right)\right)\le\max_{p\in\mathcal{I}\cup\mathcal{J}}\left(\frac{\rho_{1:m,p}}{r^{k_{2}}}+\tau_{1:m,p}\delta\right).
\]

By the definitions of $\rho_{1:m}$ and $\tau_{1:m}$, we further
have
\[
\max_{p\in\mathcal{I}\cup\mathcal{J}}\left(\frac{\rho_{1:m,p}}{r^{k_{2}}}+\tau_{1:m,p}\delta\right)\le\frac{\rho_{1:m}}{r^{k_{2}}}+\tau_{1:m}\delta.
\]

Combining the above inequalities gives, for every $k_{2}\in\mathcal{K}$,
\[
\varPhi^{k_{1}}\left(\mathcal{X}^{k_{1}}_{1:m}\right)\le\frac{\rho_{1:m}}{r^{k_{2}}}+\tau_{1:m}\delta.
\]

Since this holds for every $k_{2}\in\mathcal{K}$, we take the minimum
over $k_{2}$ and obtain
\[
\varPhi^{k_{1}}\left(\mathcal{X}^{k_{1}}_{1:m}\right)\le\min_{k\in\mathcal{K}}\left(\frac{\rho_{1:m}}{r^{k}}+\tau_{1:m}\delta\right).
\]

Using $\min_{k\in\mathcal{K}}\frac{1}{r^{k}}=\frac{1}{r_{\max}},$
we get
\[
\varPhi^{k_{1}}\left(\mathcal{X}^{k_{1}}_{1:m}\right)\le\frac{\rho_{1:m}}{r_{\max}}+\tau_{1:m}\delta.
\]

Because $k_{1}$ is an arbitrary nonempty core, and an empty core
has zero prefix workload measure, taking the maximum over all cores
gives
\[
\max_{k\in\mathcal{K}}\varPhi^{k}\left(\mathcal{X}^{k}_{1:m}\right)\le\frac{\rho_{1:m}}{r_{\max}}+\tau_{1:m}\delta.
\]

This completes the proof.
\end{IEEEproof}

\subsubsection{Group Bound for Circuit Scheduling}

\begin{lemma}[Group Scheduling Bound]
	\label{scheduling-phase-group-bound}
	For every nonempty group $G_\ell$,
	Algorithm~\ref{alg:grouped-scheduling} satisfies
	$B_\ell\le2\max_{k\in\mathcal K}W_\ell^k
	\le4(R/r_{\max}+K)b_\ell$.
\end{lemma}

\begin{IEEEproof}
	Fix a core $k$ with at least one flow in $G_\ell$,
	and let $f=f_m^k(i,j)$ be a flow that finishes last
	on this core within the group.
	Let $\theta_f=\delta+d_m^k(i,j)/r^k$ denote its total reconfiguration
	and transmission duration, and let
	$u_f=t_m^k(i,j)-s_\ell$ denote its reconfiguration start time
	relative to the group start time $s_\ell$.
	All flows in $G_\ell$ have been released by $s_\ell$.
	
	The scheduler scans all unstarted flows at time $s_\ell$
	and after each flow completion event.
	If $f$ remains unstarted after a scan, at least one of
	its ingress port $i$ and egress port $j$ must be occupied;
	otherwise, the scheduler would have started its reconfiguration.
	Thus, the waiting time $u_f$ is at most the sum of the
	busy times of ingress port $i$ and egress port $j$
	during this waiting period.
	These ports are occupied only by flows in $G_\ell$
	undergoing reconfiguration or transmission. There is no traffic from other groups, and unstarted flows
	do not reserve idle ports.
	
Both $W_{\ell,i}^k$ and $W_{\ell,j}^k$ include the
occupancy duration $\theta_f$ of $f$ itself.
Before $f$ starts, the total busy time is therefore at most
$W_{\ell,i}^k-\theta_f$ on ingress port $i$ and at most
$W_{\ell,j}^k-\theta_f$ on egress port $j$.
	Consequently,
	\[
	u_f\le W_{\ell,i}^k+W_{\ell,j}^k-2\theta_f.
	\]
	
	The completion time of $f$ relative to $s_\ell$
	therefore satisfies
	\[
	u_f+\theta_f
	\le W_{\ell,i}^k+W_{\ell,j}^k-\theta_f
	\le2W_\ell^k.
	\]
	
	A core with no assigned flows in the group finishes
	at $s_\ell$.
	Taking the maximum over all cores gives
	\[
	B_\ell\le2\max_{k\in\mathcal K}W_\ell^k.
	\]
	
	To bound the assigned workloads, let
	$q_\ell=\max\{m:\widetilde T_m\le b_\ell\}$.
	Since $G_\ell\subseteq L_{q_\ell}$ and all assigned
	workloads are nonnegative,
	Lemma~\ref{allocation-phase-prefix-bound} gives
	\[
	\max_{k\in\mathcal K}W_\ell^k
	\le
	\max_{k\in\mathcal K}
	\varPhi^k(\mathcal X_{1:q_\ell}^k)
	\le
	\frac{\rho_{1:q_\ell}}{r_{\max}}
	+\delta\tau_{1:q_\ell}.
	\]
	
	By Lemma~\ref{transmission-capacity-prefix-bound},
	$\rho_{1:q_\ell}\le2R\widetilde T_{q_\ell}$.
	Lemma~\ref{reconfiguration-capacity-prefix-bound} gives
	$\tau_{1:q_\ell}\delta\le2K\widetilde T_{q_\ell}$;
	this inequality also holds trivially when $\delta=0$.
	Thus,
	\[
	\begin{aligned}
		\max_{k\in\mathcal K}W_\ell^k
		&\le
		2\left(\frac{R}{r_{\max}}+K\right)
		\widetilde T_{q_\ell}\le
		2\left(\frac{R}{r_{\max}}+K\right)b_\ell.
	\end{aligned}
	\]
	Combining this inequality with
	$B_\ell\le2\max_{k\in\mathcal K}W_\ell^k$
	proves the lemma.
\end{IEEEproof}

\subsubsection{Derivation of Approximation Ratios (OCS Cores)}
\begin{thm}
	\label{theorem1}
	Algorithm~\ref{alg:alg1} achieves a $(32K+2)$-approximation
	for minimizing the total weighted CCT in a heterogeneous
	$K$-core OCS network with arbitrary release times.
\end{thm}

\begin{IEEEproof}
	Let $c=4(R/r_{\max}+K)$.
	By Lemma~\ref{scheduling-phase-group-bound},
	the duration of every nonempty group satisfies
	$B_h\le cb_h$.	Since $b_h=2^hb_0$, the cumulative group duration satisfies
	\[
	\sum_{h=0}^{\ell}B_h
	\le c\sum_{h=0}^{\ell}b_h
	=c(2^{\ell+1}-1)b_0
	<2cb_\ell.
	\]
	
	Fix a nonempty group $G_\ell$.
	For each nonempty group $G_h$ with $h\le\ell$,
	the LP release constraints imply
	\[
	A_h=\max_{m\in G_h}a_m
	\le\max_{m\in G_h}\widetilde T_m
	\le b_h\le b_\ell.
	\]
	Thus, by time $b_\ell$, all coflows in these groups
	have been released.
	The algorithm starts each group as soon as its coflows
	have been released and the preceding group has completed.
	Consequently, after time $b_\ell$, no further waiting
	for coflow releases is required before $G_\ell$ completes.
	The remaining processing time is at most
	$\sum_{h=0}^{\ell}B_h$, giving
	\[
	F_\ell
	\le b_\ell+\sum_{h=0}^{\ell}B_h
	\le(1+2c)b_\ell.
	\]
	
	For each coflow $m\in G_\ell$, its completion time satisfies
	$T_m\le F_\ell$.
	The group definitions give $b_\ell\le2\widetilde T_m$:
	for $\ell=0$, $b_0=\widetilde T_m$;
	for $\ell\ge1$,
	$b_\ell=2b_{\ell-1}<2\widetilde T_m$.
	Therefore,
	\[
	T_m
	\le(2+4c)\widetilde T_m
	=
	\left[
	16\left(\frac{R}{r_{\max}}+K\right)+2
	\right]\widetilde T_m.
	\]
	
	Let $T_m^\ast$ denote the completion time of coflow
	$C_m$ in an optimal schedule.
	Multiplying by $w_m$, summing over all coflows,
	and applying the LP lower bound yields
	\[
	\begin{aligned}
		\sum_m w_mT_m
		&\le
		\left[
		16\left(\frac{R}{r_{\max}}+K\right)+2
		\right]
		\sum_m w_m\widetilde T_m\\
		&\le
		\left[
		16\left(\frac{R}{r_{\max}}+K\right)+2
		\right]
		\sum_m w_mT_m^\ast.
	\end{aligned}
	\]
	Finally, $R/r_{\max}\le K$, establishing the
	$(32K+2)$-approximation guarantee.
\end{IEEEproof}

\begin{Corollary}
	\label{cor:ocs-zero-release}
	When all coflows are released at time zero,
	Algorithm~\ref{alg:alg1} achieves a $32K$-approximation
	for minimizing the total weighted CCT in a heterogeneous
	$K$-core OCS network.
\end{Corollary}

\begin{IEEEproof}
	Since all coflows are available at time zero,
	groups are processed consecutively without waiting
	for releases.
	Using the geometric sum bound from the proof of
	Theorem~\ref{theorem1}, we obtain
	\[
	F_\ell=\sum_{h=0}^{\ell}B_h
	\le8\left(\frac{R}{r_{\max}}+K\right)b_\ell.
	\]
	
	For each $m\in G_\ell$, the inequality
	$b_\ell\le2\widetilde T_m$ gives
	\[
	T_m\le F_\ell
	\le16\left(\frac{R}{r_{\max}}+K\right)\widetilde T_m
	\le32K\widetilde T_m.
	\]
	Multiplying by $w_m$, summing over all coflows,
	and applying the LP lower bound proves the corollary.
\end{IEEEproof}

\subsubsection{Derivation of Approximation Ratios (EPS Cores)}
The framework in Algorithm~\ref{alg:alg1} also applies
to multi-core EPS networks, which require no circuit
reconfiguration.
Constraints~\eqref{eq:ip-reconfig} are removed from
the LP relaxation, and flow allocation considers
only transmission workload.
The EPS variant retains LP grouping and uses
Algorithm~\ref{alg:grouped-scheduling} with $\delta=0$.

Consider a heterogeneous $H$-core EPS network
with core set $\mathcal H=\{1,\ldots,H\}$ and
per-port rates $\{r^h\}_{h\in\mathcal H}$.
Let $R=\sum_{h\in\mathcal H}r^h$ and
$r_{\max}=\max_{h\in\mathcal H}r^h$.
The prefix transmission workload measure is
\[
\varPsi^h(\mathcal Y^h_{1:m})
\triangleq
\max_{p\in\mathcal I\cup\mathcal J}
\frac{\rho^h_{1:m,p}}{r^h},
\]
where $\mathcal Y^h_{1:m}$ records the cumulative
traffic assigned to each port on core $h$
for the first $m$ coflows.

Let $\widehat T_m$ denote the LP completion time
of coflow $C_m$ in an optimal solution to the EPS
relaxation, and let $\overline T_m^\ast$ denote
its completion time in an optimal EPS schedule.
Coflows are sorted and reindexed in nondecreasing
order of $\widehat T_m$.
Set $b_0=\min_{m\in\mathcal M}\widehat T_m$
and $b_\ell=2^\ell b_0$ for $\ell\ge0$.
Groups are defined as in the OCS setting,
with $\widehat T_m$ replacing $\widetilde T_m$.

\begin{thm}
	\label{theorem2}
	The EPS variant of Algorithm~\ref{alg:alg1}
	achieves a $(16H+2)$-approximation for minimizing
	the total weighted CCT in a heterogeneous
	$H$-core EPS network with arbitrary release times.
\end{thm}

\begin{IEEEproof}
	The argument in
	Lemma~\ref{transmission-capacity-prefix-bound}
	applies to the LP relaxation (EPS variant) and gives
	$\rho_{1:m}\le2R\widehat T_m$.
	Applying the argument in
	Lemma~\ref{allocation-phase-prefix-bound}
	with the reconfiguration term removed yields
	\[
	\max_{h\in\mathcal H}
	\varPsi^h(\mathcal Y^h_{1:m})
	\le\frac{\rho_{1:m}}{r_{\max}}
	\le\frac{2R}{r_{\max}}\widehat T_m.
	\]
	
	Applying the argument of
	Lemma~\ref{scheduling-phase-group-bound}
	with port occupancy duration $d_m^h(i,j)/r^h$
	and using the prefix workload bound yields
	\[
	B_\ell\le\frac{4R}{r_{\max}}b_\ell.
	\]
	Following the proof of Theorem~\ref{theorem1}
	with $c=4R/r_{\max}$, the completion time
	$\overline T_m$ produced by the EPS variant satisfies
	\[
	\overline T_m
	\le \left(\frac{16R}{r_{\max}}+2\right)\widehat T_m.
	\]
	Multiplying by $w_m$, summing over all coflows,
	and applying the EPS LP lower bound yields
	\[
	\begin{aligned}
		\sum_m w_m\overline T_m
		&\le
		\left(\frac{16R}{r_{\max}}+2\right)
		\sum_m w_m\widehat T_m\\
		&\le
		\left(\frac{16R}{r_{\max}}+2\right)
		\sum_m w_m\overline T_m^\ast.
	\end{aligned}
	\]
	Finally, $R/r_{\max}\le H$, establishing the
	$(16H+2)$-approximation guarantee.
\end{IEEEproof}

\begin{Corollary}
	\label{cor:eps-zero-release}
	When all coflows are released at time zero,
	the EPS variant achieves a $16H$-approximation
	in a heterogeneous $H$-core EPS network.
\end{Corollary}

\begin{IEEEproof}
	With zero release times, the argument of
	Theorem~\ref{theorem2} gives
	$\overline T_m\le16(R/r_{\max})\widehat T_m$,
	which, together with the LP lower bound and
	$R/r_{\max}\le H$, establishes the
	$16H$-approximation guarantee.
\end{IEEEproof}

Table~\ref{tab:approximation-guarantees} summarizes
the approximation guarantees for arbitrary and zero
release times.

\begin{table}[htp]
	\centering
	\caption{Approximation Guarantees}
	\label{tab:approximation-guarantees}
	\small
	\setlength{\tabcolsep}{4pt}
	\renewcommand{\arraystretch}{1.2}
	\begin{tabular}{lcc}
		\toprule
		Network & Arbitrary release times & Zero release times\\
		\midrule
		OCS ($K$ cores) & $32K+2$ & $32K$\\
		EPS ($H$ cores) & $16H+2$ & $16H$\\
		\bottomrule
	\end{tabular}
\end{table}

\section{Experimental Evaluations}
\label{sec:Experimental-Evaluations}
In this section, we evaluate the proposed algorithm through
trace-driven simulations using a Facebook workload.

\subsection{Experimental Setup}
\textbf{Workload:}
We derive the workload from the Facebook trace
\cite{facebook}, collected from a MapReduce cluster
with 3000 machines and 150 racks.
The benchmark contains 526 coflows represented on
a 150-port fabric and has been widely used in
coflow scheduling studies
\cite{sunflow,regularization,wang2024scheduling,
	wang2023efficient,wang2023online,wang2025optimal,huang2020weaver}.
For each coflow, the trace records sender and receiver
locations and the aggregate traffic received by each
receiver, rather than individual flow sizes.

To construct the demand matrices, we distribute
each receiver's aggregate traffic equally among
the coflow's senders.
We independently perturb each flow size by a percentage
sampled uniformly from $[-10\%,10\%]$.
We then rescale the flow sizes for each receiver
to preserve its aggregate traffic.
We randomly select $N$ distinct racks and map them
to both the ingress and egress ports of an $N$-port network.

\textbf{Default Parameters:}
Unless otherwise specified, we use $N=10$ ingress
and $N=10$ egress ports, $M=100$ coflows, and $K=3$
OCS cores.
The per-port transmission rates are
$(r^1,r^2,r^3)=(10,20,30)$ Gbps,
giving an aggregate port transmission rate
of $R=60$ Gbps.
The reconfiguration delay is $\delta=100$ ms.
Coflow weights are independently sampled as
$w_m\sim\mathcal U\{1,\ldots,10\}$.
For each random seed, all algorithms use the same
demand matrices, weights, and release times.
Each experiment uses 30 independent random seeds,
and we report means with 95\% confidence intervals.

\textbf{Performance Metrics:}
Our primary objective is to minimize the total weighted
CCT $\sum_{m=1}^{M}w_mT_m$.
We compare the algorithms using the normalized total
weighted CCT:
\[
\mathrm{NormW}(\mathcal A)
\triangleq
\frac{\sum_{m=1}^{M}w_mT_m(\mathcal A)}
{\sum_{m=1}^{M}w_mT_m(\textsc{Ours})},
\]
where $T_m(\mathcal A)$ denotes the completion time
of coflow $C_m$ under algorithm $\mathcal A$, and
\textsc{Ours} denotes Algorithm~\ref{alg:alg1},
including the circuit scheduling procedure in
Algorithm~\ref{alg:grouped-scheduling}.
A value greater than one indicates that $\mathcal A$
incurs a higher total weighted CCT than \textsc{Ours}.

We also measure the gap between the total weighted CCT
of \textsc{Ours} and the LP lower bound:
\[
\mathrm{LPGap}
\triangleq
\frac{\sum_{m=1}^{M}w_mT_m(\textsc{Ours})}
{\sum_{m=1}^{M}w_m\widetilde T_m},
\]
where $\widetilde T_m$ denotes the LP completion time
of coflow $C_m$. A smaller LP gap indicates performance closer to the lower bound.

\subsection{Baseline Solutions}
\begin{itemize}
	\item \textbf{\textsc{Greedy-Approx}:}
	This baseline adopts the approximation algorithm
	proposed in \cite{wang2026scheduling} for coflows
	released simultaneously.
	We extend it to different release times by allowing
	only released flows to start and updating scheduling
	decisions at coflow release and flow completion events.
	
	\item \textbf{\textsc{Load-Only}:}
	This baseline retains the LP coflow order and circuit
	scheduling of \textsc{Ours}, but considers only
	transmission workload during inter-core flow allocation.
	Each flow is assigned to the core minimizing the maximum
	cumulative port transmission workload after its
	tentative assignment, excluding the reconfiguration
	term $\delta\tau_{1:m,p}^k$.
	Circuit scheduling still includes the delay $\delta$.
	
	\item \textbf{\textsc{Fastest-Core}:}
	This baseline assigns all flows to the fastest OCS core,
	$k^\ast=\arg\max_{k\in\mathcal K}r^k$,
	leaving the other cores unused.
	It retains the LP coflow order and circuit scheduling
	of \textsc{Ours}.
	
	\item \textbf{\textsc{WSPT-Order}:}
	This baseline replaces the LP priority order
	of \textsc{Ours} with the weighted shortest processing
	time (WSPT) rule, ordering coflows by
	$w_m/T_{\mathrm{LB}}(D_m)$ in nonincreasing order.
	Here, $T_{\mathrm{LB}}(D_m)=\delta+\rho_m/R$
	is a lower bound on the time required to complete
	$C_m$ when scheduled alone \cite{wang2026scheduling}.
	It retains the flow allocation rule and LP-based
	grouped circuit scheduling of \textsc{Ours}.
\end{itemize}

\subsection{Simulation Results}
\subsubsection{Performance under the Default Setting}
Fig.~\ref{fig:weighted_unweighted_cct} compares the
normalized total weighted and unweighted CCT under
the default setting.
\textsc{Greedy-Approx} yields weighted and unweighted ratios
of 1.952 and 1.621, respectively.
\textsc{WSPT-Order} yields ratios of 1.190 and 1.177,
respectively, showing the benefit of LP-guided priorities
for flow allocation and circuit scheduling.
\textsc{Load-Only} yields ratios of 1.409 and 1.384,
showing the benefit of accounting for reconfiguration
workload during inter-core flow allocation.
\textsc{Fastest-Core} yields ratios of 2.323 and 2.327,
illustrating the benefit of distributing traffic
across multiple OCS cores.

\begin{figure}[htp]
    \centering
    \includegraphics[width=0.8\columnwidth]{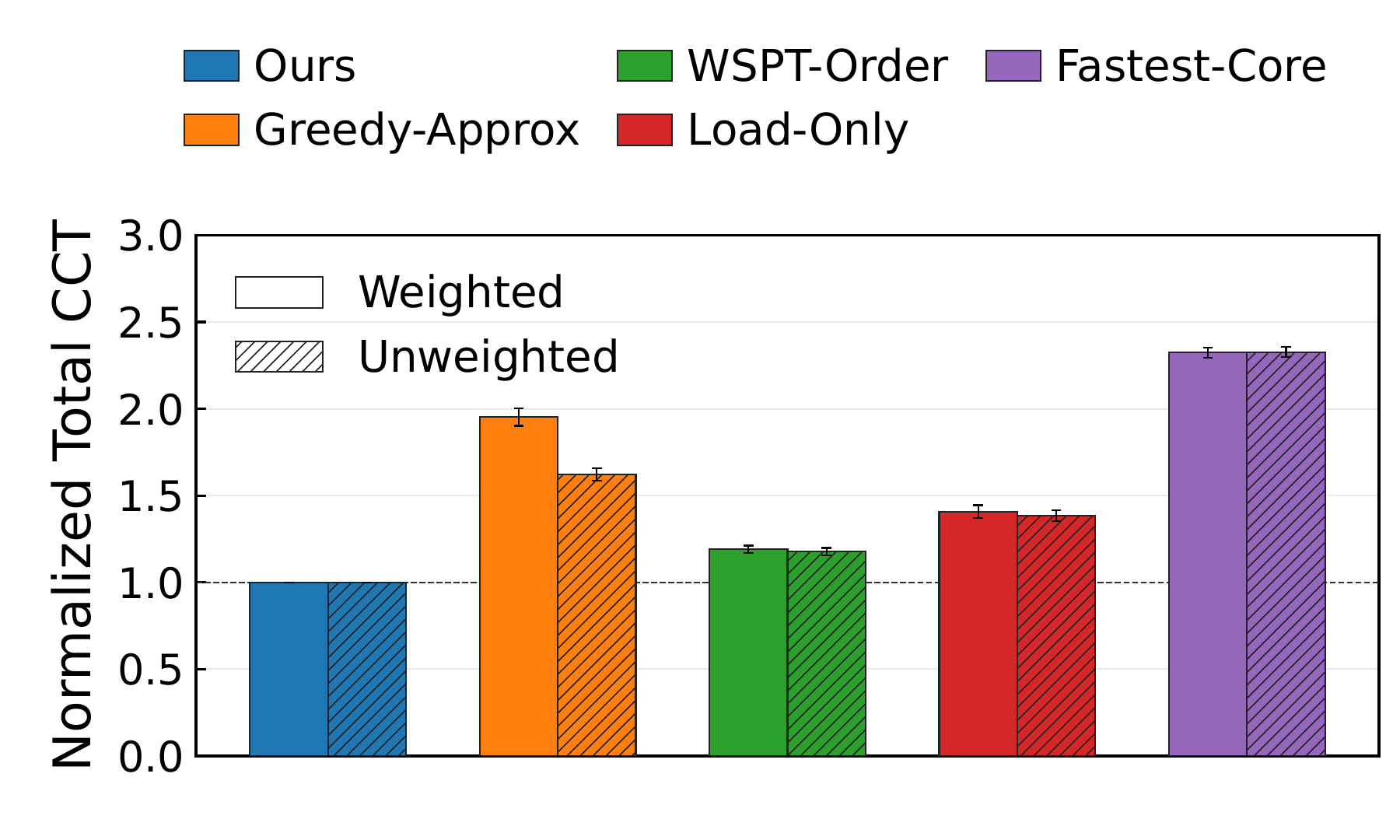}
    \caption{Normalized total weighted and unweighted CCT under the default setting.}
    \label{fig:weighted_unweighted_cct}
\end{figure}

\subsubsection{Impact of the Number of Cores}
Fig.~\ref{fig:core_sensitivity} evaluates the impact of
the number of OCS cores under a fixed aggregate port
transmission rate.
We fix $N=10$, $M=100$, $R=60$ Gbps, and
$\delta=100$ ms, and vary $K\in\{1,2,3,4,5\}$.
The corresponding per-port rate vectors are
$[60]$, $[20,40]$, $[10,20,30]$,
$[6,12,18,24]$, and $[4,8,12,16,20]$ Gbps, respectively.

The left subfigure compares total weighted CCT
normalized by \textsc{Ours} at the same value of $K$.
The ratio of \textsc{Greedy-Approx} decreases from
2.271 at $K=1$ to 1.814 at $K=5$.
\textsc{WSPT-Order} is close to \textsc{Ours} at $K=1$,
with a mean ratio of 1.002, while its ratio
increases to 1.281 at $K=5$.
For \textsc{Load-Only}, the ratio increases from
1.253 at $K=2$ to 1.553 at $K=5$,
showing the benefit of reconfiguration-aware allocation.
The ratio of \textsc{Fastest-Core} increases from
1.724 to 3.331, as it leaves
the other cores unused.

\begin{figure}[htp]
    \centering
    \includegraphics[width=\columnwidth]{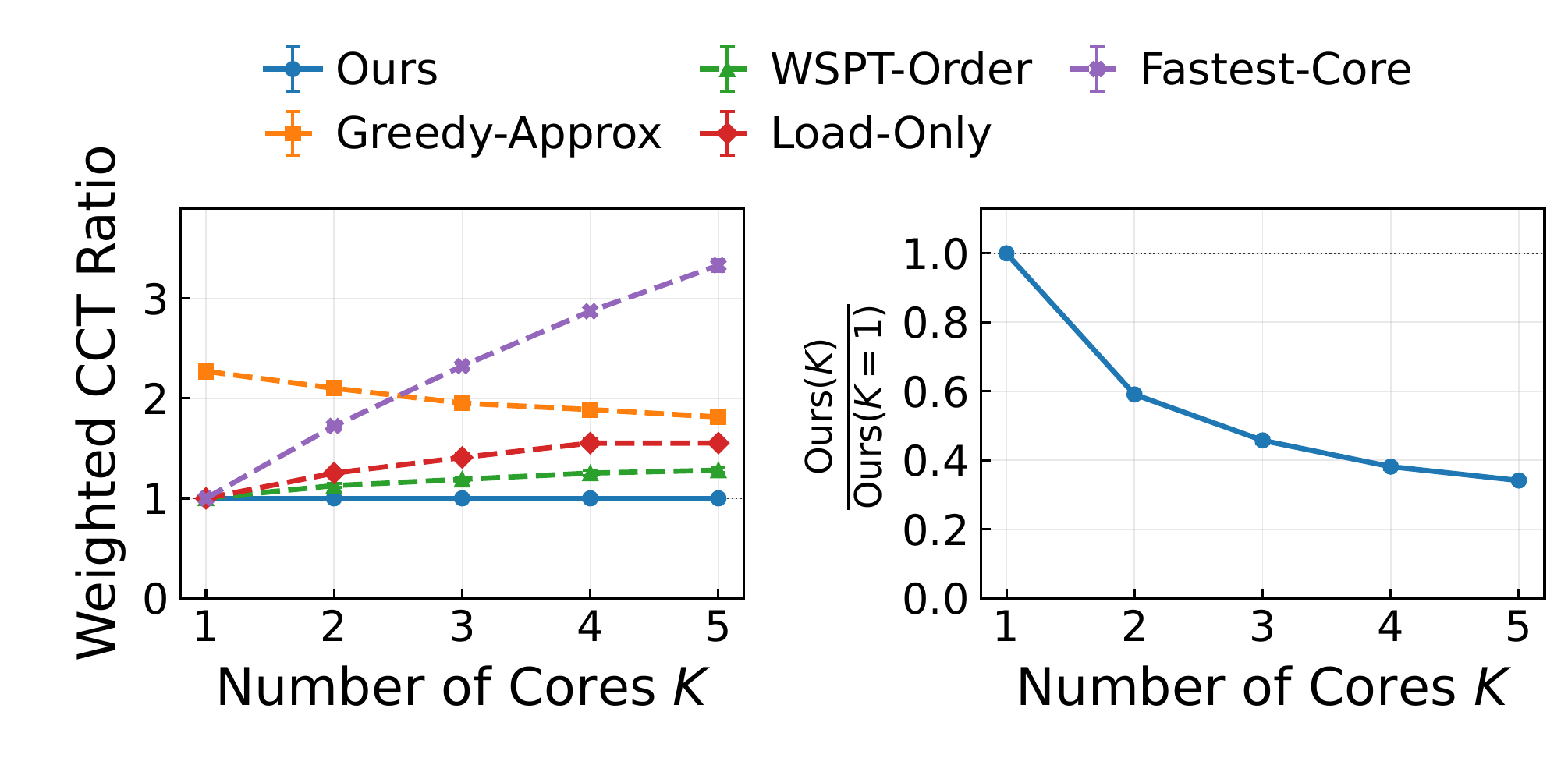}
    \caption{Impact of $K$ under a fixed $R=60$ Gbps.}
    \label{fig:core_sensitivity}
\end{figure}

The right subfigure shows the total weighted CCT
of \textsc{Ours} normalized by its value at $K=1$.
The mean ratio decreases monotonically from
1.000 at $K=1$ to 0.591, 0.458, 0.382, and 0.341
for $K=2,3,4,5$, respectively.
These results show that additional OCS cores reduce
total weighted CCT even when aggregate bandwidth
remains fixed.
Multiple cores allow circuits sharing a logical
ingress or egress port to operate concurrently
on different cores and distribute reconfiguration
workload across cores.

\subsubsection{Impact of Reconfiguration Delay}
We vary the reconfiguration delay over
$\delta\in\{20,50,100,150,200\}$ ms, including the
default value of 100 ms, while fixing $N=10$,
$M=100$, and $R=60$ Gbps.
For $K=3,4,5$, we consider both heterogeneous
and homogeneous core rates.
The heterogeneous rate vectors are $[10,20,30]$,
$[6,12,18,24]$, and $[4,8,12,16,20]$ Gbps,
respectively.
The corresponding homogeneous rate vectors are
$[20,20,20]$, $[15,15,15,15]$, and
$[12,12,12,12,12]$ Gbps, respectively.

\begin{itemize}
	\item \textbf{$K=3$:}
	Fig.~\ref{fig:delta_sensitivity_K3} shows that
	\textsc{WSPT-Order} remains closer to \textsc{Ours}
	than \textsc{Greedy-Approx}.
	Their ratios are 1.144--1.212 and 1.889--1.966
	under heterogeneous rates, respectively, and
	1.134--1.223 and 1.838--1.964 under homogeneous rates.
	As $\delta$ increases from 20 to 200 ms,
	the ratio of \textsc{Load-Only} rises from 1.238
	to 1.445 under heterogeneous rates and from 1.229
	to 1.442 under homogeneous rates.
	For \textsc{Fastest-Core}, the ratio under heterogeneous
	rates increases from 2.113 to 2.373, while that under
	homogeneous rates changes only slightly, from 2.436 to 2.467.
\end{itemize}

\begin{figure}[htp]
	\centering
	\includegraphics[width=\columnwidth]{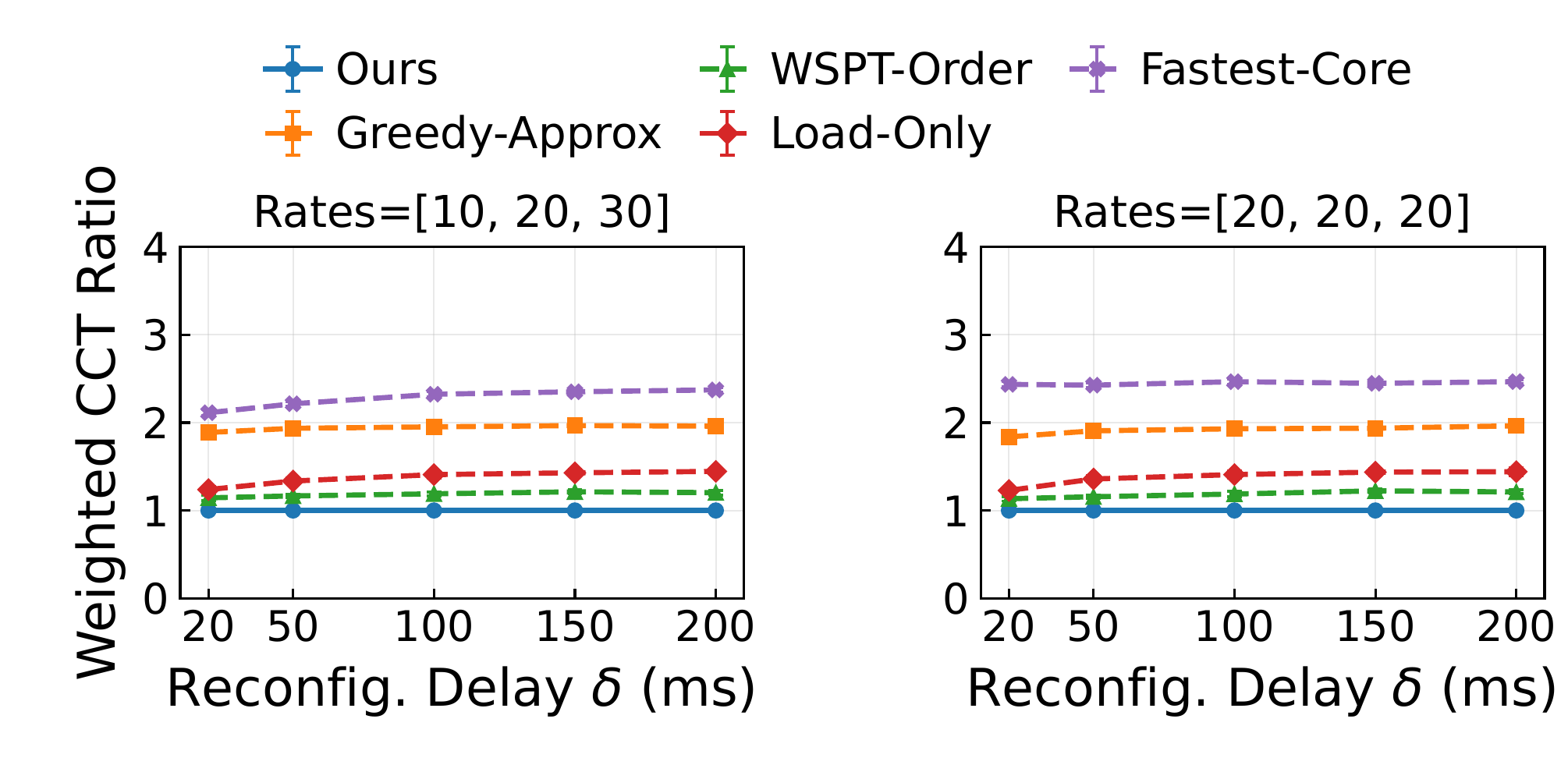}
	\caption{Impact of $\delta$ under heterogeneous
		and homogeneous core rates when $K=3$.}
	\label{fig:delta_sensitivity_K3}
\end{figure}

\begin{itemize}
	\item \textbf{$K=4$:}
	Fig.~\ref{fig:delta_sensitivity_K4} presents the results.
	As $\delta$ increases from 20 to 200 ms,
	the ratio of \textsc{Load-Only} changes from
	1.279 to 1.569 under heterogeneous rates and
	from 1.255 to 1.580 under homogeneous rates.
	The corresponding ratios of \textsc{Fastest-Core}
	change from 2.599 to 2.932 and from 3.143 to 3.072,
	respectively.
	\textsc{WSPT-Order} yields ratios of 1.166--1.292
	under heterogeneous rates and 1.155--1.272
	under homogeneous rates.
	For \textsc{Greedy-Approx}, the corresponding ranges
	are 1.873--1.904 and 1.800--1.878.
\end{itemize}

\begin{figure}[htp]
	\centering
	\includegraphics[width=\columnwidth]{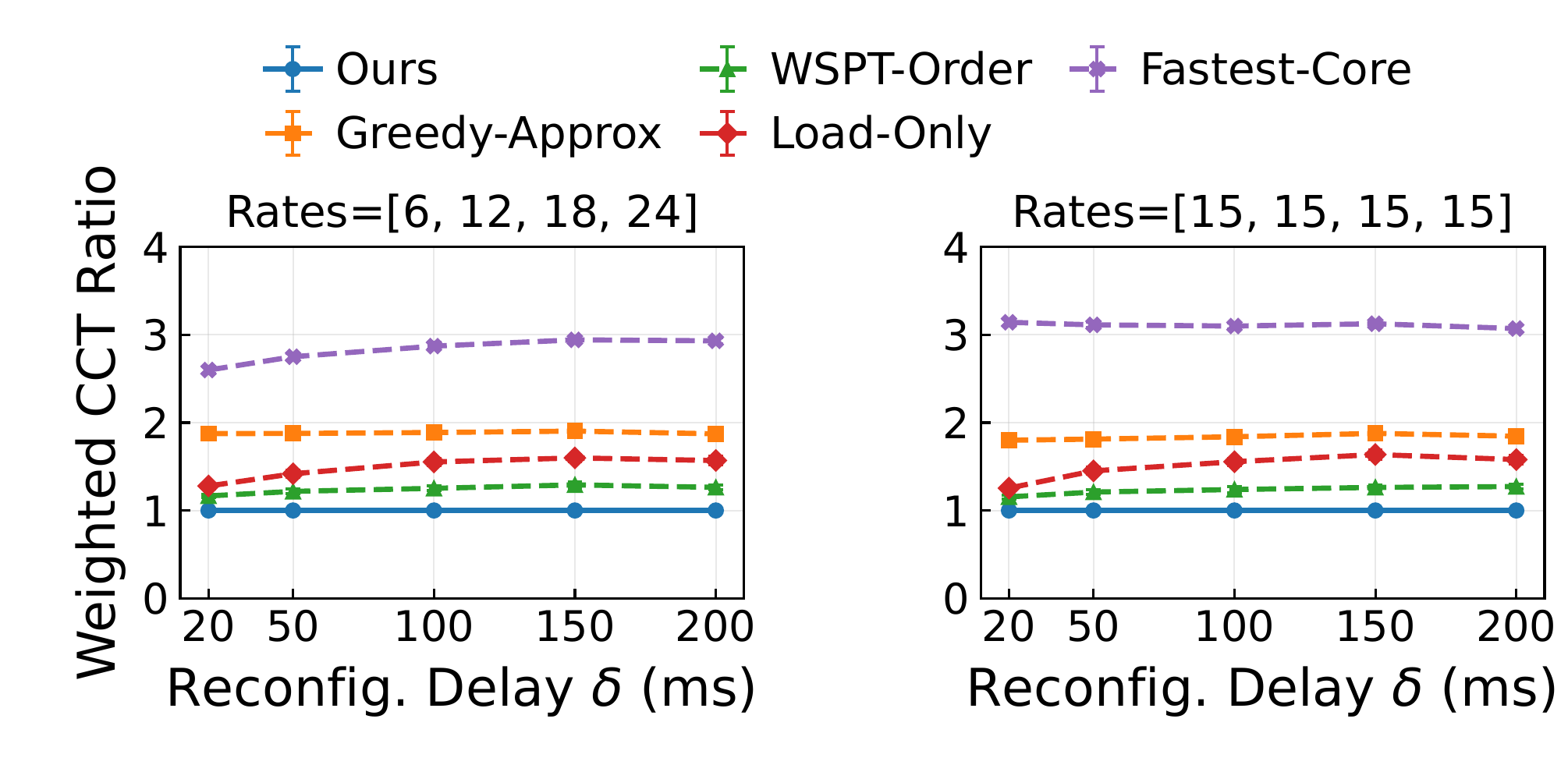}
	\caption{Impact of $\delta$ under heterogeneous
		and homogeneous core rates when $K=4$.}
	\label{fig:delta_sensitivity_K4}
\end{figure}

\begin{figure}[htp]
	\centering
	\includegraphics[width=\columnwidth]{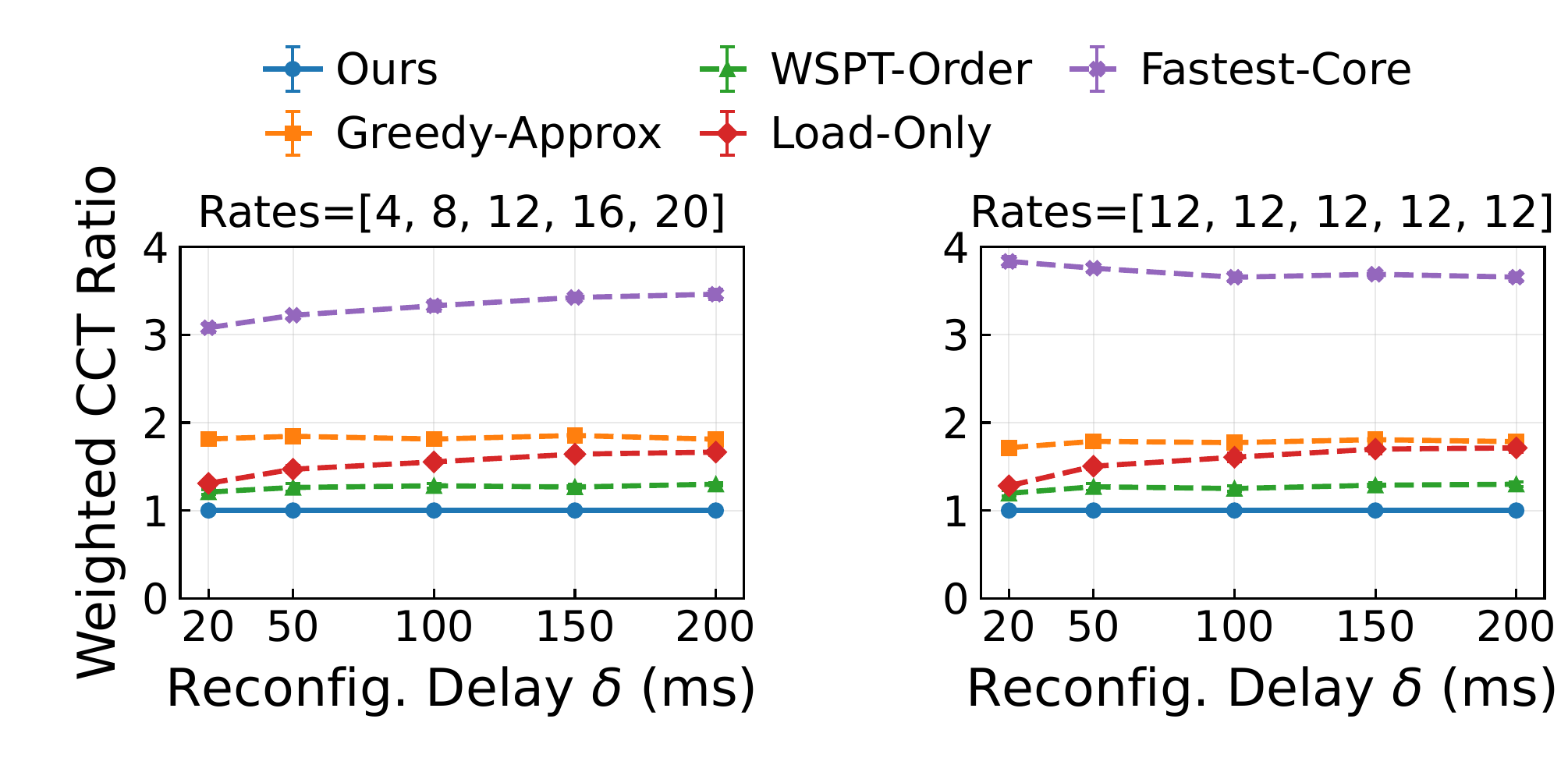}
	\caption{Impact of $\delta$ under heterogeneous
		and homogeneous core rates when $K=5$.}
	\label{fig:delta_sensitivity_K5}
\end{figure}

\begin{figure*}[htp]
	\centering
	\includegraphics[width=0.75\textwidth]{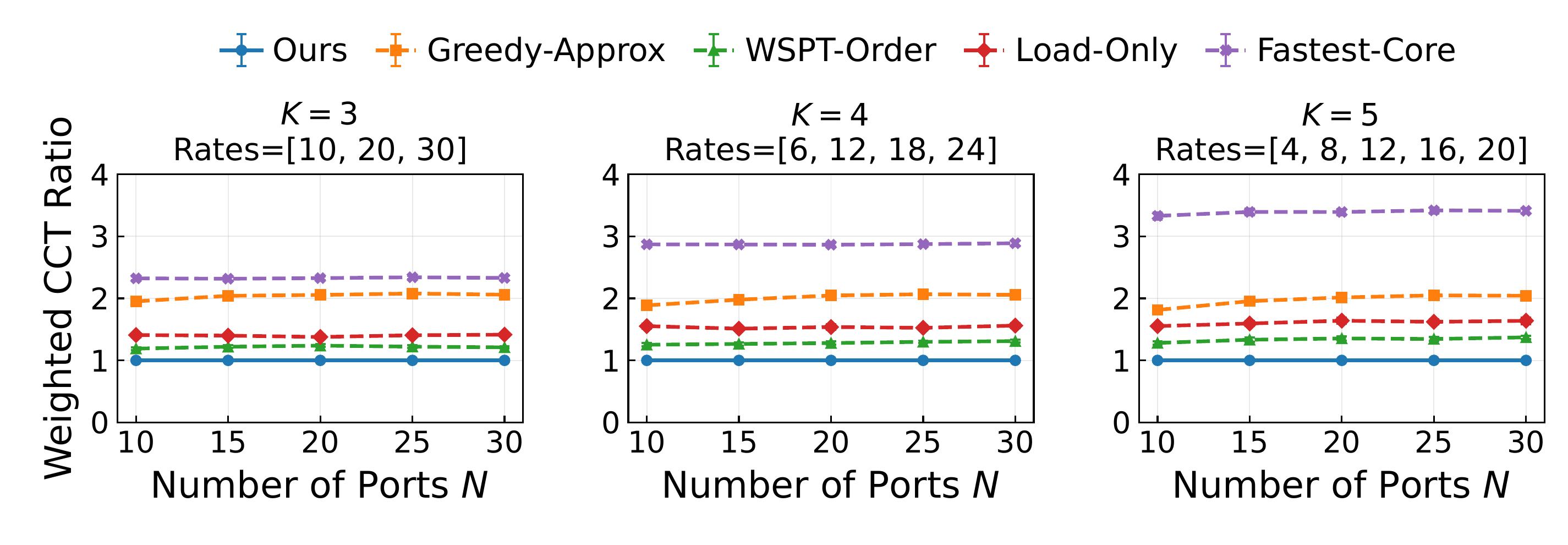}
	\caption{Impact of $N$ under heterogeneous core rates.}
	\label{fig:n_sensitivity}
\end{figure*}

\begin{figure*}[htp]
	\centering
	\includegraphics[width=0.75\textwidth]{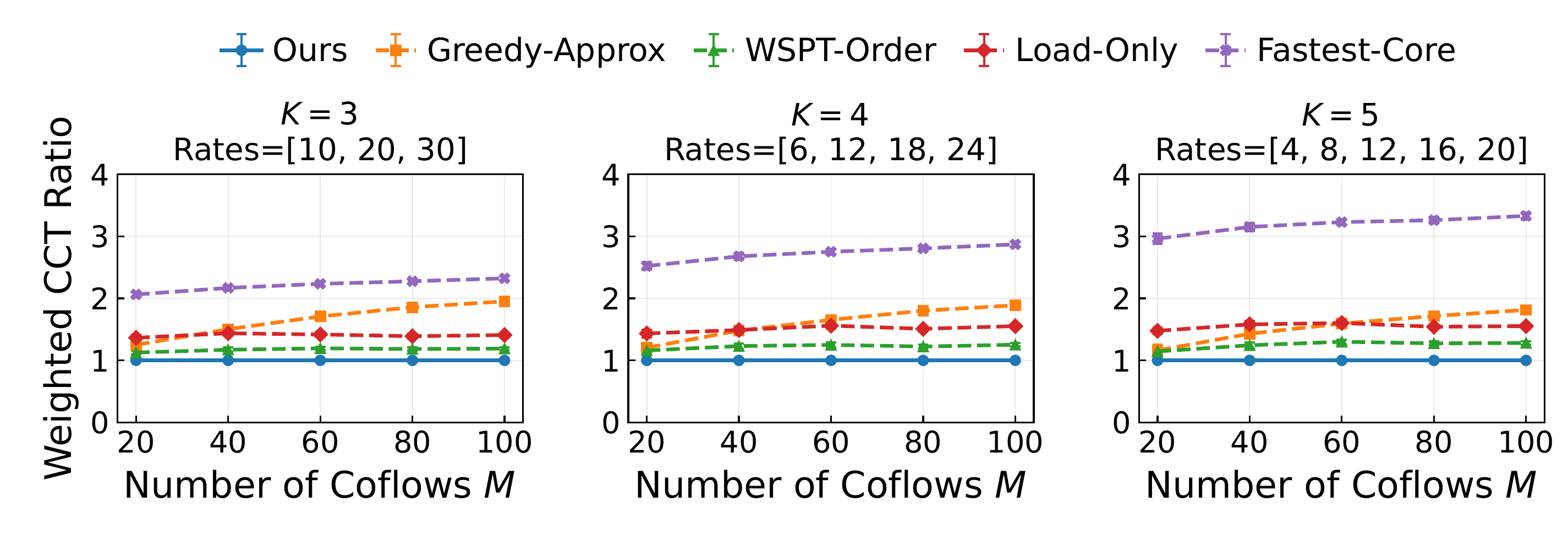}
	\caption{Impact of $M$ under heterogeneous core rates.}
	\label{fig:m_sensitivity}
\end{figure*}

\begin{itemize}
	\item \textbf{$K=5$:}
	Fig.~\ref{fig:delta_sensitivity_K5} further illustrates
	the cost of ignoring reconfiguration workload
	during allocation.
	Between $\delta=20$ and $200$ ms, the ratio of
	\textsc{Load-Only} increases from 1.309 to 1.665
	under heterogeneous rates and from 1.281 to 1.713
	under homogeneous rates.
	\textsc{Fastest-Core} incurs the largest ratios,
	changing from 3.081 to 3.463 and from 3.836 to 3.657
	in the two settings, respectively.
	By comparison, the ratios of \textsc{WSPT-Order}
	remain within 1.209--1.300 under heterogeneous rates
	and 1.196--1.299 under homogeneous rates.
	The corresponding ranges for \textsc{Greedy-Approx}
	are 1.812--1.854 and 1.715--1.807.
\end{itemize}

\subsubsection{Impact of the Number of Ports}
We vary the number of ingress and egress ports over
$N\in\{10,15,20,25,30\}$ while fixing
$M=100$, $R=60$ Gbps, and $\delta=100$ ms.
For $K=3,4,5$, the heterogeneous per-port rate vectors
are $[10,20,30]$, $[6,12,18,24]$, and
$[4,8,12,16,20]$ Gbps, respectively.

Fig.~\ref{fig:n_sensitivity} compares total weighted
CCT normalized by \textsc{Ours} at the same $N$ and $K$.
Across the tested configurations, the mean ratios of
\textsc{Greedy-Approx}, \textsc{WSPT-Order}, and
\textsc{Load-Only} are 1.814--2.079, 1.190--1.371,
and 1.378--1.640, respectively.
\textsc{Fastest-Core} incurs larger penalties
as the number of cores increases, with ratios of
2.317--2.341, 2.865--2.889, and 3.331--3.420
for $K=3,4,5$, respectively.
These results illustrate the benefit of distributing
traffic across the available cores.

\subsubsection{Impact of the Number of Coflows}
We vary the number of coflows over $M\in\{20,40,60,80,100\}$
while fixing $N=10$, $R=60$ Gbps, and $\delta=100$ ms.
For $K=3,4,5$, the heterogeneous per-port rate vectors
are $[10,20,30]$, $[6,12,18,24]$, and
$[4,8,12,16,20]$ Gbps, respectively.

Fig.~\ref{fig:m_sensitivity} reports total weighted CCT
normalized by \textsc{Ours} at the same $M$ and $K$.
As $M$ increases from 20 to 100, the mean ratio of
\textsc{Greedy-Approx} increases from 1.251 to 1.952,
from 1.204 to 1.888, and from 1.161 to 1.814
for $K=3,4,5$, respectively.
\textsc{WSPT-Order} yields mean ratios of
1.126--1.194, 1.160--1.252, and 1.146--1.300
for $K=3,4,5$, respectively.
\textsc{Load-Only} yields mean ratios of 1.364--1.603,
while \textsc{Fastest-Core} yields ratios of 2.063--3.331.
These comparisons show the benefits of LP-guided priorities,
reconfiguration-aware allocation, and traffic distribution
across cores over the tested workload sizes.

\begin{figure}[htp]
	\centering
	\includegraphics[width=\columnwidth]{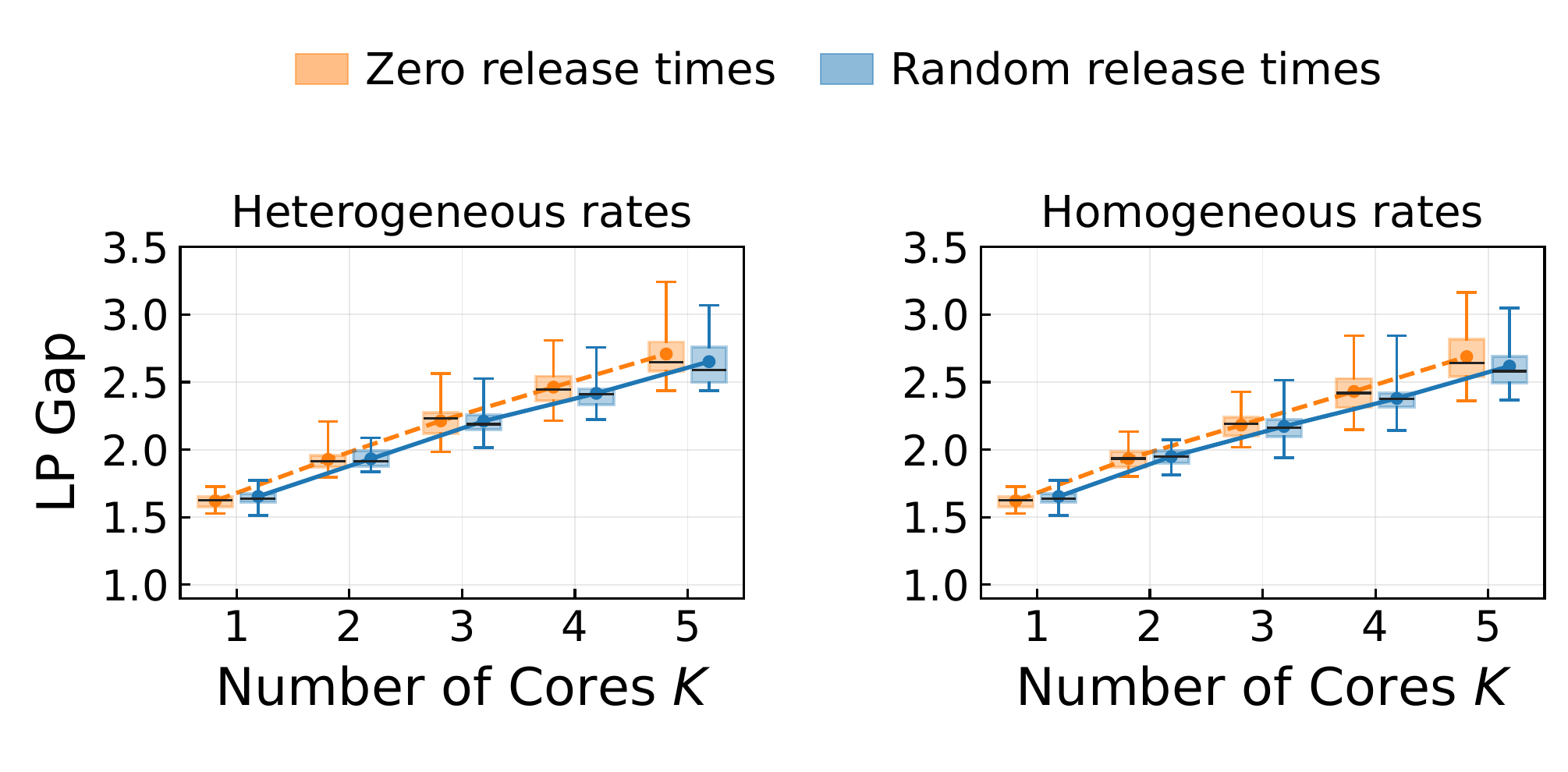}
	\caption{LP gaps under heterogeneous and homogeneous core rates.}
	\label{fig:approx_core}
\end{figure}

\subsubsection{Gap to the LP Lower Bound}
We fix $N=10$, $M=100$, $R=60$ Gbps, and
$\delta=100$ ms, and vary $K\in\{1,2,3,4,5\}$.
The heterogeneous per-port rate vectors are $[60]$,
$[20,40]$, $[10,20,30]$, $[6,12,18,24]$, and
$[4,8,12,16,20]$ Gbps, respectively.
The corresponding homogeneous vectors are $[60]$,
$[30,30]$, $[20,20,20]$, $[15,15,15,15]$, and
$[12,12,12,12,12]$ Gbps.

Fig.~\ref{fig:approx_core} shows the distributions of LP gaps
over 30 random seeds.
Boxes indicate interquartile ranges, horizontal lines mark
medians, whiskers show sample ranges, and connected circles
denote means. As $K$ increases from 1 to 5, the mean LP gap under heterogeneous
rates increases from 1.621 to 2.708 for zero release times
and from 1.654 to 2.652 for random release times.
Under homogeneous rates, it increases from 1.621 to 2.688
for zero release times and from 1.654 to 2.619
for random release times.
All mean LP gaps remain well below the worst-case approximation bounds.

\section{Conclusions\label{sec:Conclusions}}
This paper develops an approximation algorithm for coflow
scheduling in heterogeneous multi-core OCS networks under
the \textit{not-all-stop} reconfiguration model.
The algorithm combines LP-guided coflow ordering,
greedy flow allocation, and grouped circuit scheduling.
It achieves approximation ratios of $32K+2$ and $32K$
for arbitrary and zero release times, respectively,
in a $K$-core OCS network.
These guarantees depend only on the number of cores,
eliminating the dependence on the number of coflows
and the weight ratio in the earlier approximation bound.
The same framework can be naturally applied to heterogeneous
$H$-core EPS networks through a transmission-only variant,
achieving approximation ratios of $16H+2$ and $16H$
for arbitrary and zero release times, respectively.
The experimental results show that the proposed algorithm
achieves lower total weighted CCT than the evaluated baselines.

Future work will investigate online coflow scheduling
in heterogeneous multi-core OCS networks under the
\textit{not-all-stop} reconfiguration model, focusing on
the design of online algorithms and the analysis of
their competitive ratios for minimizing total weighted CCT.

\section*{Acknowledgement}

This work is supported by the Department of Environment, Science and Innovation
of Queensland State Government under Quantum 2032 Challenge Program
(Project \#Q2032001). The corresponding author is Hong Shen.

\bibliographystyle{IEEEtran}
\addcontentsline{toc}{section}{\refname}\bibliography{IEEE_Journal}

@Article{Baraat,
  author    = {Dogar, Fahad R and Karagiannis, Thomas and Ballani, Hitesh and Rowstron, Antony},
  journal   = {ACM SIGCOMM Computer Communication Review},
  title     = {Decentralized task-aware scheduling for data center networks},
  year      = {2014},
  number    = {4},
  pages     = {431--442},
  volume    = {44},
  publisher = {ACM New York, NY, USA},
}

@Article{Aalo,
  author    = {Chowdhury, Mosharaf and Stoica, Ion},
  journal   = {ACM SIGCOMM Computer Communication Review},
  title     = {Efficient coflow scheduling without prior knowledge},
  year      = {2015},
  number    = {4},
  pages     = {393--406},
  volume    = {45},
  publisher = {ACM New York, NY, USA},
}

@Misc{facebook,
  howpublished = {\url{https://github.com/coflow/coflow-benchmark}},
  title        = {FaceBookTrace},
  year         = {2019},
}

@Article{mapreduce,
  author    = {Dean, Jeffrey and Ghemawat, Sanjay},
  journal   = {Communications of the ACM},
  title     = {MapReduce: simplified data processing on large clusters},
  year      = {2008},
  number    = {1},
  pages     = {107--113},
  volume    = {51},
  publisher = {ACM New York, NY, USA},
}

@InProceedings{spark,
  author    = {Zaharia, Matei and Chowdhury, Mosharaf and Das, Tathagata and Dave, Ankur and Ma, Justin and McCauly, Murphy and Franklin, Michael J and Shenker, Scott and Stoica, Ion},
  booktitle = {9th $\{$USENIX$\}$ Symposium on Networked Systems Design and Implementation ($\{$NSDI$\}$ 12)},
  title     = {Resilient distributed datasets: A fault-tolerant abstraction for in-memory cluster computing},
  year      = {2012},
  pages     = {15--28},
}

@InProceedings{dryad,
  author    = {Isard, Michael and Budiu, Mihai and Yu, Yuan and Birrell, Andrew and Fetterly, Dennis},
  booktitle = {Proceedings of the 2nd ACM SIGOPS/EuroSys European Conference on Computer Systems 2007},
  title     = {Dryad: distributed data-parallel programs from sequential building blocks},
  year      = {2007},
  pages     = {59--72},
}

@InProceedings{networking,
  author    = {Chowdhury, Mosharaf and Stoica, Ion},
  booktitle = {Proceedings of the 11th ACM Workshop on Hot Topics in Networks},
  title     = {Coflow: A networking abstraction for cluster applications},
  year      = {2012},
  pages     = {31--36},
}

@InProceedings{NC-DRF,
  author       = {Wang, Luping and Wang, Wei},
  booktitle    = {2018 IEEE 38th International Conference on Distributed Computing Systems (ICDCS)},
  title        = {Fair coflow scheduling without prior knowledge},
  year         = {2018},
  organization = {IEEE},
  pages        = {22--32},
}

@article{literature13,
	title={Scheduling dependent coflows to minimize the total weighted job completion time in datacenters},
	author={Tian, Bingchuan and Tian, Chen and Wang, Bingquan and Li, Bo and He, Zehao and Dai, Haipeng and Liu, Kexin and Dou, Wanchun and Chen, Guihai},
	journal={Computer Networks},
	volume={158},
	pages={193--205},
	year={2019},
	publisher={Elsevier}
}

@inproceedings{literature9,
	title={Minimizing the total weighted completion time of coflows in datacenter networks},
	author={Qiu, Zhen and Stein, Cliff and Zhong, Yuan},
	booktitle={Proceedings of the 27th ACM symposium on Parallelism in Algorithms and Architectures},
	pages={294--303},
	year={2015}
}

@article{literature32,
	title={Efficient scheduling of weighted coflows in data centers},
	author={Wang, Zhiliang and Zhang, Han and Shi, Xingang and Yin, Xia and Li, Yahui and Geng, Haijun and Wu, Qianhong and Liu, Jianwei},
	journal={IEEE Transactions on Parallel and Distributed Systems},
	volume={30},
	number={9},
	pages={2003--2017},
	year={2019},
	publisher={IEEE}
}

@inproceedings{literature6,
	title={Efficient coflow scheduling with varys},
	author={Chowdhury, Mosharaf and Zhong, Yuan and Stoica, Ion},
	booktitle={Proceedings of the 2014 ACM conference on SIGCOMM},
	pages={443--454},
	year={2014}
}

@inproceedings{literature10,
	title={Brief announcement: Improved approximation algorithms for scheduling co-flows},
	author={Khuller, Samir and Purohit, Manish},
	booktitle={Proceedings of the 28th ACM Symposium on Parallelism in Algorithms and Architectures},
	pages={239--240},
	year={2016}
}

@InProceedings{sunflow,
  author    = {Huang, Xin Sunny and Sun, Xiaoye Steven and Ng, TS Eugene},
  booktitle = {Proceedings of the 12th International on Conference on emerging Networking EXperiments and Technologies},
  title     = {Sunflow: Efficient optical circuit scheduling for coflows},
  year      = {2016},
  pages     = {297--311},
}

@InProceedings{co-scheduler,
  author       = {Li, Zhuozhao and Shen, Haiying},
  booktitle    = {2019 IEEE 39th International Conference on Distributed Computing Systems (ICDCS)},
  title        = {Co-scheduler: Accelerating data-parallel jobs in datacenter networks with optical circuit switching},
  year         = {2019},
  organization = {IEEE},
  pages        = {186--195},
}

@Article{zhang2020minimizing,
  author    = {Zhang, Tong and Ren, Fengyuan and Bao, Jiakun and Shu, Ran and Cheng, Wenxue},
  journal   = {IEEE Transactions on Parallel and Distributed Systems},
  title     = {Minimizing coflow completion time in optical circuit switched networks},
  year      = {2020},
  number    = {2},
  pages     = {457--469},
  volume    = {32},
  publisher = {IEEE},
}

@InProceedings{omco,
  author       = {Xu, Chao and Tan, Haisheng and Hou, Jiahui and Zhang, Chi and Li, Xiang-Yang},
  booktitle    = {2018 IEEE International Conference on Communications (ICC)},
  title        = {OMCO: Online multiple coflow scheduling in optical circuit switch},
  year         = {2018},
  organization = {IEEE},
  pages        = {1--6},
}

@Article{BvN,
  author  = {Birkhoff, Garrett},
  journal = {Univ. Nac. Tucuman, Ser. A},
  title   = {Tres observaciones sobre el algebra lineal},
  year    = {1946},
  pages   = {147--154},
  volume  = {5},
}

@Article{regularization,
  author    = {Tan, Haisheng and Zhang, Chi and Xu, Chao and Li, Yupeng and Han, Zhenhua and Li, Xiang-Yang},
  journal   = {IEEE/ACM Transactions on Networking},
  title     = {Regularization-based coflow scheduling in optical circuit switches},
  year      = {2021},
  number    = {3},
  pages     = {1280--1293},
  volume    = {29},
  publisher = {IEEE},
}

@article{improved,
  author={Shafiee, Mehrnoosh and Ghaderi, Javad},
  journal={IEEE/ACM Transactions on Networking},
  title={An improved bound for minimizing the total weighted completion time of coflows in datacenters},
  volume={26},
  number={4},
  pages={1674--1687},
  year={2018},
  publisher={IEEE}
}

@Article{theoretical5,
  author    = {Shafiee, Mehrnoosh and Ghaderi, Javad},
  journal   = {IEEE/ACM Transactions on Networking},
  title     = {Scheduling coflows with dependency graph},
  year      = {2021},
  number    = {1},
  pages     = {450--463},
  volume    = {30},
  publisher = {IEEE},
}

@InProceedings{decentralized1,
  author       = {Luo, Shouxi and Yu, Hongfang and Zhao, Yangming and Wu, Bin and Wang, Sheng and others},
  booktitle    = {2015 IEEE International Conference on Communications (ICC)},
  title        = {Minimizing average coflow completion time with decentralized scheduling},
  year         = {2015},
  organization = {IEEE},
  pages        = {307--312},
}

@InProceedings{ONS,
  author       = {Jiang, Renjie and Zhang, Tong and Yi, Changyan},
  booktitle    = {2023 IEEE Symposium on Computers and Communications (ISCC)},
  title        = {Effective Coflow Scheduling in Hybrid Circuit and Packet Switching Networks},
  year         = {2023},
  organization = {IEEE},
  pages        = {1156--1161},
}

@InProceedings{reco,
  author       = {Zhang, Chi and Tan, Haisheng and Xu, Chao and Li, Xiang-Yang and Tang, Shaojie and Li, Yupeng},
  booktitle    = {2019 IEEE 39th International Conference on Distributed Computing Systems (ICDCS)},
  title        = {Reco: Efficient regularization-based coflow scheduling in optical circuit switches},
  year         = {2019},
  organization = {IEEE},
  pages        = {111--121},
}

@Article{wang2023efficient,
  author    = {Wang, Xin and Shen, Hong and Tian, Hui},
  journal   = {IEEE Transactions on Network and Service Management},
  title     = {Efficient and Fair: Information-Agnostic Online Coflow Scheduling by Combining Limited Multiplexing with DRL},
  year      = {2023},
  number    = {4},
  pages     = {4572--4584},
  volume    = {20},
  publisher = {IEEE},
}

@Article{wang2023online,
  author    = {Wang, Xin and Shen, Hong},
  journal   = {Future Generation Computer Systems},
  title     = {Online scheduling of coflows by attention-empowered scalable deep reinforcement learning},
  year      = {2023},
  pages     = {195--206},
  volume    = {146},
  publisher = {Elsevier},
}

@article{wang2024scheduling,
  title={Scheduling coflows in hybrid optical-circuit and electrical-packet switches with performance guarantee},
  author={Wang, Xin and Shen, Hong and Tian, Hui},
  journal={IEEE/ACM Transactions on Networking},
  volume={32},
  number={3},
  pages={2299--2314},
  year={2024},
  publisher={IEEE}
}

@article{gonzalez1976open,
  title={Open shop scheduling to minimize finish time},
  author={Gonzalez, Teofilo and Sahni, Sartaj},
  journal={Journal of the ACM (JACM)},
  volume={23},
  number={4},
  pages={665--679},
  year={1976},
  publisher={ACM New York, NY, USA}
}

@article{chen2023scheduling,
  title={Scheduling coflows for minimizing the total weighted completion time in heterogeneous parallel networks},
  author={Chen, Chi-Yeh},
  journal={Journal of Parallel and Distributed Computing},
  volume={182},
  pages={104752},
  year={2023},
  publisher={Elsevier}
}

@inproceedings{huang2020weaver,
  title={Weaver: Efficient coflow scheduling in heterogeneous parallel networks},
  author={Huang, Xin Sunny and Xia, Yiting and Ng, TS Eugene},
  booktitle={2020 IEEE International Parallel and Distributed Processing Symposium (IPDPS)},
  pages={1071--1081},
  year={2020},
  organization={IEEE}
}

@article{wang2025optimal,
  title={Optimal Partitioning of Traffic Demand for Coflow Scheduling in Hybrid Switches},
  author={Wang, Xin and Shen, Hong and Tian, Hui},
  journal={IEEE Transactions on Network and Service Management},
  year={2025},
  publisher={IEEE}
}

@misc{cisco2016gci,
  author       = {{Cisco. (2016)}},
  title        = {Cisco Global Cloud Index: Forecast and Methodology, 2015--2020},
  howpublished = {\url{https://www.cisco.com/c/dam/en/us/solutions/collateral/service-provider/global-cloud-index-gci/white-paper-c11-738085.pdf}},
}

@misc{cisco2016_40g,
  author       = {{Cisco White Paper. (2016)}},
  title        = {The Future is 40 Gigabit Ethernet},
  howpublished = {\url{https://www.cisco.com/c/dam/en/us/products/collateral/switches/catalyst-6500-series-switches/white-paper-c11-737238.pdf}},
}

@article{chen2023efficient,
  title={Efficient approximation algorithms for scheduling coflows with total weighted completion time in identical parallel networks},
  author={Chen, Chi-Yeh},
  journal={IEEE Transactions on Cloud Computing},
  volume={12},
  number={1},
  pages={116--129},
  year={2023},
  publisher={IEEE}
}

@inproceedings{poutievski2022jupiter,
  title={Jupiter evolving: transforming google's datacenter network via optical circuit switches and software-defined networking},
  author={Poutievski, Leon and Mashayekhi, Omid and Ong, Joon and Singh, Arjun and Tariq, Mukarram and Wang, Rui and Zhang, Jianan and Beauregard, Virginia and Conner, Patrick and Gribble, Steve and others},
  booktitle={Proceedings of the ACM SIGCOMM 2022 Conference},
  pages={66--85},
  year={2022}
}

@article{wang2026scheduling,
	title         = {Scheduling Coflows in Multi-Core OCS Networks with Performance Guarantee},
	author        = {Wang, Xin and Shen, Hong and Tian, Hui and Wang, Dong},
	journal       = {arXiv preprint arXiv:2604.08242},
	year          = {2026},
	eprint        = {2604.08242},
	archivePrefix = {arXiv},
	primaryClass  = {cs.DC},
	doi           = {10.48550/arXiv.2604.08242},
	url           = {https://doi.org/10.48550/arXiv.2604.08242}
}

@inproceedings{wang2026coflow,
	title={Coflow Scheduling in Hybrid-Switched Data Center Networks under Not-All-Stop Reconfiguration},
	author={Wang, Xin and Shen, Hong and Tian, Hui},
	booktitle={Proceedings of the 55th International Conference on Parallel Processing},
	pages={1008--1018},
	year={2026}
}

\vspace{12pt}

\end{document}